\documentclass{article}
\usepackage{amsmath}
\usepackage{amssymb}
\usepackage{amsthm}
\usepackage{tikz}
\usepackage{amssymb, stmaryrd, extarrows}
\usepackage{caption,subcaption}
\usepackage{algorithm2e}
\usepackage{wrapfig}
\usepackage{paralist}
\usepackage{times}
\usepackage{mdwlist}
\usepackage[caption=false]{subfig} 

\newcommand*{\pr}{\mathbb{P}}

\newcommand*{\expect}{\mathbb{E}}
\newtheorem{theorem}{Theorem}
\newtheorem{remark}[theorem]{Remark}
\newtheorem{lemma}[theorem]{Lemma}
\newtheorem{assumption}[theorem]{Assumption}
\newtheorem{definition}[theorem]{Definition}
\newtheorem{corollary}[theorem]{Corollary}

\newcommand\restr[2]{\ensuremath{\left.#1\right|_{#2}}}
\newcommand\supp{\textrm{\sf Supp}}
\newcommand\almostsure{\textrm{\sf AS}}

\newcommand\abso{\textsf{Abs}}

\newcommand\objreach{\textrm{\sf Reach}}

\newcommand\Wsafety{\textrm{\sf Safe}}
\newcommand\red{\textsf{red}}
\newcommand\calT{\mathcal{T}}
\newcommand\calH{\mathcal{H}}
\newcommand\calP{\mathcal{P}}
\newcommand\calD{\mathcal{D}}
\newcommand\calG{\mathcal{G}}
\newcommand\calM{\mathcal{M}}

\newcommand\hatabstr{\hat{\mathcal{A}}}
\newcommand\wtabstr{\wt{\mathcal{A}}}
\newcommand\barabstr{\bar{\mathcal{A}}}

\newcommand\Inf{\textrm{\sf Inf}}
\newcommand\mecs{\textrm{\sf MEC}}

\newcommand\Val{\ensuremath{{\textrm{\sf Val}}}}
\newcommand\wt[1]{\widetilde{#1}}
\usetikzlibrary{arrows,shapes,snakes,automata,calc}

\title{Multiple-Environment \\
Markov Decision Processes}
\author{Jean-Fran\c cois Raskin \and Ocan Sankur}

\begin{document}
\maketitle
\begin{abstract}
We introduce Multi-Environment Markov Decision Processes (MEMDPs) which are MDPs with a \emph{set} of probabilistic transition functions. The goal in a MEMDP is to synthesize a single controller with guaranteed performances against
\emph{all} environments even though the environment is unknown a priori. While MEMDPs can be seen as a special class of partially observable MDPs, we show that several verification problems that are undecidable for partially observable MDPs, are decidable for MEMDPs and sometimes have even efficient solutions.  
\end{abstract}

\section{Introduction}
\emph{Markov decision processes (MDP)} are a standard formalism for modeling systems that exhibit both
stochastic and non-deterministic aspects. 
At each round of the execution of a MDP, an action is chosen
by a controller (resolving the non-determinism), and the next state is determined
by a probability distribution associated to the current state and the chosen action.
A controller is thus a \emph{strategy} (a.k.a. \emph{policy}) that determines which action to choose
at each round according to the history of the execution so far.
Algorithms for finite state MDPs are known for a large variety of objectives including omega-regular objectives~\cite{CY-acm95}, PCTL objectives~\cite{BK-book08}, or quantitative objectives~\cite{Puterman-wiley94}.

\paragraph{Multiple-Environment MDP}
In a MDP, the environment is {\em unique}, and this may not be realistic: we may
want to design a control strategy that exhibits good performances under several
hypotheses formalized by different models for the environment, and those environments may {\em not} be distinguishable or we may {\em not want} to distinguish them (e.g. because it is too costly to design several control strategies.)
As an illustration, consider the design of guidelines for a medical treatment
that needs to work adequately for two populations of patients (each given by a different stochastic model), even if the patients cannot be diagnosed to be in one population or in the other. 
A appropriate model for this case would be a MDP with two different models for the responses of the patients to the sequence of actions taken during the cure. We want a therapy that possibly takes decisions by observing the reaction of the patient and that works well (say reaches a good state for the patient with high probability) no matter if the patient belongs to the first of the second population. 

Facing two potentially indistinguishable environments can be easily modelled
with a partially observable MDPs. Unfortunately, this model is particularly
intractable~\cite{DBLP:conf/csl/ChatterjeeCT13} (e.g. qualitative and
quantitative reachability, safety and parity objectives are undecidable.) To remedy to this situation, we introduce \emph{multiple-environment MDPs (MEMDP)} which are MDPs with a \emph{set} of probabilistic transition functions, rather than a {\em single one}. 
The goal in a MEMDP is to synthesize a single controller with guaranteed performances against
\emph{all} environments even though the environment at play is unknown a priori (it may be discovered during interaction but not necessarily.)
We show that verification problems that are undecidable for partially observable MDPs, are decidable for MEMDPs and sometimes have even efficient solutions.


\paragraph{Results}
We study MEMDPs with three types of objectives: reachability, safety and parity
objectives. For each of those objectives, we study both {\em qualitative} and
{\em quantitative} threshold decision problems~\footnote{For readability, we
concentrate in this paper on MEMDPs with two environments, but most of the results
can be easily generalized to any finite number of environments possibly with an increased computational complexity. This is left for a long version of this paper.}. We first show that winning strategies may need infinite memory as well as randomization, and we provide algorithms to solve the decision problems.  As it is classical, we consider two variants for the qualitative threshold problems. The first variant, asks to determine the existence of a {\em single} strategy that wins the objective with probability one (almost surely winning) in all the environments of the MEMDP. The second variant asks to determine the existence of a family of {\em single} strategies such that for all $\epsilon > 0$, there is one strategy in the family that wins the objective with probability larger than $1-\epsilon$ (limit sure winning) in all the environments of the MEMDP. For both almost sure winning and limit sure winning, and for all three types of objectives, we provide efficient polynomial time algorithmic solutions. Then we turn to the quantitative threshold problem that asks for the existence of a single strategy that wins the objective with a probability that exceeds a given rational threshold in all the environments. We show the problem to be NP-hard (already for two environments and acyclic MEMDPs), and so classical quantitative analysis techniques based on LP cannot be applied here. 
Instead, we show that finite memory strategies are sufficient to approach
achievable thresholds and we reduce the existence of bounded memory  strategies
to solving quadratic equations,
leading to solutions in polynomial space.
    
%

\paragraph{Related Work}
In addition to partially observable MDPs, our work is related to the following research lines.

Interval Markov chains are Markov chains in which transition probabilities are only known to belong
to given intervals (see \textit{e.g.} \cite{kozine2002interval,KS-fsttcs05,CHK-ipl13}).
Similarly, Markov decision processes with uncertain transition
matrices for finite-horizon and discounted cases were considered
\cite{nilim2005robust}. 
The latter work also mentions the finite scenario-case which is similar to our setting.
However, the precise distributions of actions at
each round are assumed to be independent while in our work we consider it to be fixed but unknown.
Independence is a {\em simplifying assumption} that only 
provides pessimistic guarantees.
However this approach does not use the information one obtains on the system
along observed histories, and so the results tend to be overly pessimistic.

Our work is related to reinforcement learning, where the goal is to develop
strategies which ensure good performance in unknown environments, by learning
and optimizing simultaneously; see \cite{kaelbling1996reinforcement} for a
survey. In particular, it is related to the multi-armed bandit problem
 where one is given a set of {\em stateless} systems with unknown reward distributions, and
the goal is to choose the best one while optimizing the overall cost incured
while learning. The problem of finding the optimal one (without optimizing) with high confidence was considered 
in \cite{EMM-clt02,MT-jmlr04}, and is related to our constructions inside
\emph{distinguishing double end-components} (see Section~\ref{section:dec}).
However, our problems differ from this one as in multi-armed bandit problem models of the bandits are unknown while our environments are known but we do not know a priori against which one we are playing. 

MEMDPs are also related to multi-objective reachability in MDPs
considered in~\cite{EKVY-lmcs08}, where a strategy is to be synthesized so as to
ensure the reachability of a set of targets, each with a possibly different
probability. If we allow multiple environments and possibly different reachability objectives
for each environment, this problem can be reduced to reachability in
MEMDPs. Note however that the general reachability problem is harder in MEMDPs;
it is NP-hard even for acyclic MEMDPs with absorbing targets, while
polynomial-time algorithms exist for absorbing targets in the setting
of~\cite{EKVY-lmcs08}.


\section{Definitions}
\label{section:definitions}
A finite~\emph{Markov decision process (MDP)} is a tuple $M = (S, A, \delta)$, where
$S$ is a finite set of \emph{states},
$A$ a finite set of \emph{actions},
and $\delta : S \times A \rightarrow \mathcal{D}(S)$ a partial function, where $\mathcal{D}(S)$ is
the set of \emph{probability distributions} on~$S$.
For any state~$s \in S$, we denote by~$A(s)$ the set of
actions available from~$s$.
We define a \emph{run} of~$M$ as a finite or infinite sequence $s_1a_1 \ldots a_{n-1}
s_n\ldots$ of
states and actions such that $\delta(s_i,a_i,s_{i+1})>0$ for all~$i\geq 1$.
Finite runs are also called \emph{histories} and denoted $\calH(M)$.

\textit{Sub-MDPs and End-components}~
For the following definitions, we fix an MDP $M = (S,A,\delta)$.
A \emph{sub-MDP $M'$ of~$M$} is
an MDP $(S',A',\delta')$ with $S' \subseteq S$, $A' \subseteq A$,
and such that for all $s \in S'$, $A'(s)\neq \emptyset$ and for all~$a \in
A'(s)$, we have $\supp(\delta(s,a)) \subseteq S'$, and $\delta'(s,a) = \delta(s,a)$.
For all subsets $S' \subseteq S$ with the property that for all $s \in S'$, there exists $a \in A(s)$
with $\supp(\delta(s,a)) \subseteq S'$,
we define the \emph{sub-MDP of~$M$ induced by $S'$} as the maximal sub-MDP
whose states are~$S'$, and denote it by $\restr{M}{S'}$. 
In other terms, the
sub-MDP induced by~$S'$ contains all actions of~$S'$ whose supports are
inside~$S'$.
An MDP is strongly connected if between any pair of states $s,t$, there is a
path.
An \emph{end-component} of $M = (S,A,\delta)$ is a sub-MDP
$M'=(S',A',\delta')$ that is strongly connected.
It is known that the union of two end components with non-empty intersection is
an end-component; one can thus define \emph{maximal} end-components.
We let $\mecs(M)$ denote the set of maximal end-components of~$M$, computable in
polynomial time~\cite{DeAlfaro-phd97}.
An \emph{absorbing state} $s$ is such that for all $a \in A(s)$, $\delta(s,a,s) =
1$. We denote by $\abso(M)$ the set of absorbing states of MDP~$M$.

\textit{Histories and Strategies}~
A~\emph{strategy}~$\sigma$ is a function $(SA)^*S\rightarrow \calD(A)$ such that
for all~$h \in (SA)^*S$ ending in~$s$, we have~$\supp(\sigma(h)) \in A(s)$.
A strategy is \emph{pure} if all histories are mapped to \emph{Dirac
distributions}.
A strategy~$\sigma$ is \emph{finite-memory} if it can be encoded with a
\emph{stochastic Moore machine}, $(\calM,\sigma_a,\sigma_u,\alpha)$ 
where~$\calM$ is a finite set of memory elements,~$\alpha$ the \emph{initial
distribution on~$\calM$}, 
$\sigma_u$ the \emph{memory update
function} $\sigma_u : A\times S \times \calM \rightarrow \calD(\calM)$, 
and $\sigma_a : S \times \calM \rightarrow \calD(A)$ the \emph{next action
function} where $\supp(\sigma(s,m)) \subseteq A(s)$ for any~$s\in S$ and~$m \in \calM$.
A~\emph{$K$-memory strategy} is such that~$|\calM|=K$.
A~\emph{memoryless strategy} is such that~$|\calM|=1$, and thus only depends on the last state of the history. We
define such strategies as functions $s \mapsto \calD(A(s))$ for ~$s \in
S$. 
An MDP~$M$, a finite-memory strategy~$\sigma$ encoded by $(\calM,\sigma_a,\sigma_u,\alpha)$, and a state~$s$ determine a
finite Markov chain $M_s^\sigma$ defined on the state space $S\times \calM$ as follows.
The initial distribution is such that for any~$m \in \calM$,
state $(s,m)$ has probability $\alpha(m)$, and~$0$ for other states. For any pair of states
$(s,m)$ and~$(s',m')$, the probability of the transition $(s,m),a,(s',m')$ is
equal to $\sigma_a(s,m)(a) \cdot \delta(s,a,s') \cdot \sigma_u(s,m,a)(m')$.
A~\emph{run} of~$M_s^\sigma$ is a finite or infinite sequence of the form
$(s_1,m_1),a_1,(s_2,m_2),a_2,\ldots$, where each
$(s_i,m_i),a_i,(s_{i+1},m_{i+1})$ is a transition with nonzero probability
in~$M_s^\sigma$, and~$s_1=s$. In this case, the run $s_1a_1 s_2a_2\ldots$, obtained by
projection to~$M$, is said to be \emph{compatible with~$\sigma$}.
When considering the probabilities of events in $M_s^\sigma$, we
will often consider sets of runs of~$M$. Thus, given $E \subseteq
(SA)^*$, we denote by $\pr_{M,s}^\sigma[E]$ the probability of the runs 
of~$M_s^\sigma$ whose projection to~$S$ is in~$E$.

For any strategy~$\sigma$ in a MDP~$M$, and a sub-MDP~$M'=(S',A',\delta')$, we say
that~$\sigma$ is \emph{compatible with~$M'$} if for any $h \in (SA)^*S'$,
$\supp(\sigma(h)) \subseteq A'(s)$.

Let~$\Inf(w)$ denote the disjoint union of states and actions that occur
infinitely often in the run~$w$; $\Inf$ is thus seen as a random variable.
By a slight abuse of notation,
we say that $\Inf(w)$ is equal to a sub-MDP~$D$ whenever it contains
exactly the states and actions of~$D$. 
It was shown that for any MDP~$M$, state~$s$, strategy~$\sigma$,
${\pr_{M,s}^\sigma[\Inf \in \mecs(M)]=1}$~\cite{DeAlfaro-phd97}.
We call a subset of states \emph{transient} if it is visited finitely many times with probability
$1$ under any strategy. 

\textit{Objectives}~
Given a set~$T$ of states, we define a \emph{safety objective w.r.t.~$T$},
written $\Wsafety(T)$, as the set of runs that only visit~$T$.
A \emph{reachability objective w.r.t. $T$}, written $\objreach(T)$, is the set of
runs that visit~$T$ at least once.
We also consider \emph{parity objectives}. A \emph{parity function} is defined
on the set of states $p : S \rightarrow \{0,1,\ldots,2d\}$ for some nonnegative
integer~$d$. 
The set of \emph{winning runs of~$M$ for~$p$} is defined as
$\calP_p = \{ w \in (SA)^\omega \mid \min\{ p(s) \mid s \in \Inf(w)\} \in
2\mathbb{N}\}$.
For any MDP~$M$, state~$s$, strategy~$\sigma$, 
and objective~$\Phi$, we denote 
$\Val_{\Phi}^\sigma(M,s) = \pr_{M,s}^\sigma[\Phi]$
and $\Val_\Phi^*(M,s) = \sup_{\sigma} \pr_{M,s}^\sigma[\Phi]$.
We say that objective~$\Phi$ is \emph{achieved surely} if for some~$\sigma$, all runs of~$M$
from~$s$ compatible with~$\sigma$ satisfy~$\Phi$.
Objective~$\Phi$ is \emph{achieved with probability~$\alpha$}
in~$M$ from~$s$ if for some~$\sigma$, $\Val_{\Phi}^\sigma(M,s)\geq \alpha$.
If~$\Phi$ is achieved with probability~$1$, we say that it is \emph{achieved
almost surely}. Objective~$\Phi$ is \emph{achieved limit-surely} if for
any~$\epsilon>0$, it is achieved with probability $1-\epsilon$.
In MDPs, limit-sure achievability coincides with almost-sure achievability since
optimal strategies exist.
We define $\almostsure(M,\Phi)$ as the set of states of~$M$ where~$\Phi$ is
achieved almost surely. Recall that for reachability, safety, and
parity objectives these states can be computed in polynomial time, and are only
dependent on the supports of the probability distributions~\cite{BK-book08,DeAlfaro-phd97}.
In particular, there exists a strategy ensuring $\Phi$ almost-surely when started from any state
of~$\almostsure(M,\Phi)$.
It is known that for any MDP~$M$, state~$s$, and a reachability, safety, or parity objective, 
there exists a pure memoryless
strategy $\sigma$ computable in polynomial time achieving the optimal
value~\cite{Puterman-wiley94,CY-acm95}.
The algorithm for parity objectives is obtained by showing that in each
end-component the probability of ensuring the objective is either~$0$ or~$1$,
and then reducing the problem to the reachability of those \emph{winning}
end-components. In the next lemma, we recall that the classification of winning end-components does
not depend on the exact values of the probabilities, but only on the support of
the distributions.

\begin{lemma}[\cite{DeAlfaro-phd97}]
	\label{lemma:mdp-mec-parity}
	Let~$M=(S,A,\delta)$ be a strongly connected MDP, and~$p$ a parity function.
	Then, for any MDP~$M'=(S,A,\delta')$ such that for all~$s \in S$, $a \in A$,
	$\supp(\delta(s,a)) = \supp(\delta'(s,a))$, and for all states~$s \in S$,
	there exists a strategy~$\sigma$ such that
	$\Val_{\calP_p}^\sigma(M,s) = \Val_{\calP_p}^*(M,s)=\Val_{\calP_p}^*(M',s) =
	\Val_{\calP_p}^\sigma(M',s)\in \{0,1\}$.
\end{lemma}

\section{Multiple-Environment MDP}
\label{section:memdp}
A \emph{multiple-environment MDP (MEMDP)}, is a tuple $M = (S, A, (\delta_i)_{1\leq
i\leq k})$, where for each~$i$, $(S,A,\delta_i)$ is an MDP.
We will denote by $M_i$ the MDP
obtained by fixing the edge probabilities $\delta_i$,
so that $\pr_{M_i,s}^\sigma[E]$ denotes the probability of event~$E$ in~$M_i$
from state~$s$ under strategy~$\sigma$.
Intuitively, each~$M_i$ corresponds to the behavior of the system at hand under
a different \emph{environment}; in fact, while the state space is identical in
each~$M_i$, the transition probabilities between states and even their supports may differ.

In this paper, for readability, we will study the case of~$k=2$.
We are interested in synthesizing a \emph{single} strategy~$\sigma$ with guarantees on
\emph{both} environments, without a priori knowing against which 
environment~$\sigma$ is playing.
We consider reachability, safety, and parity objectives, and again for
readability, we consider the case where the same objective is to hold in all
environments. 
The general quantitative problem is the following.
\begin{definition}
	\label{definition:memdp-problem}
	Given a MEMDP~$M$, state~$s_0$, rationals $\alpha_1,\alpha_2$,
	and objective $\Phi$, which is a reachability, safety, or a parity objective,
	compute a strategy~$\sigma$, if it exists, such that
	\(
		\forall i =1,2, \Val_{\Phi}^\sigma(M_i,s) \geq \alpha_i.
		\)
\end{definition}

We refer to the general problem as \emph{quantitative reachability (resp.
safety, parity)}.
For an instance~$M$, $s_0$, $(\alpha_1,\alpha_2)$, $\Phi$, 
we say that $\Phi$ is \emph{achieved with probabilities $(\alpha_1,\alpha_2)$}
in~$M$ from~$s$ if there is a strategy~$\sigma$ witnessing the above definition.
We say that $\Phi$ is \emph{achieved almost
surely} in~$M$ from~$s$ if it is achieved with probabilities~$(1,1)$.
Objective~$\Phi$ is \emph{achieved
limit-surely} in~$M$ from~$s$ if for any~$\epsilon>0$,
$\Phi$ is achieved in~$M$ from~$s$ with probabilities $(1-\epsilon,1-\epsilon)$.
\emph{Almost-sure reachability (resp. safety, parity)} problems consist in
deciding whether in a given~$M$, from a state~$s$, a given objective is achieved
almost surely.
\emph{Limit-sure reachability (resp. safety, parity)} 
problems are defined respectively.
Note that in MDPs and MEMDPs, almost-sure safety coincides with \emph{sure}
safety (requiring that all runs compatible with a given strategy stay in the safe set of states).

\textit{Strategy Complexity~}
We note that unlike MDPs, all considered objectives may require infinite
memory and randomization, and Pareto-optimal probability vectors may not be achievable
(a Pareto-optimal vector is componentwise maximal).
All counterexamples are given in Fig.~\ref{fig:counterexamples}.

\begin{lemma}
	For some MEMDPs~$M$ and reachability objectives~$\Phi$:
	\begin{itemize}
		\item 
			there exists a randomized strategy that achieves~$\Phi$ with higher
			probabilities in both environments than any pure strategy,
		\item there exists an infinite-memory strategy that achieves~$\Phi$ with higher
			probabilities in both environments than any finite-memory strategy,
		\item objective~$\Phi$ can be achieved limit-surely but not almost surely
		(showing Pareto-optimal vectors are not always achievable).
	\end{itemize}	
%
%
\end{lemma}

The first item is clear from Fig.~\ref{fig:counterexamples}, while the second item follows from the results
of the paper. The third item is implies by the next lemma. 

\begin{lemma}
	\label{lemma:inf-memory-reach}
	In the MEMDP~$M$ of Fig.~\ref{fig:inf-memory-reach}, for the reachability
	objective $\objreach(T)$, there exists a Pareto-optimal
	vector of probabilities achievable by an infinite-memory strategy but not by
	any finite-memory strategy.
\end{lemma}
\begin{proof}
	Clearly, $u$ is almost surely reached under any strategy. 
	Let us denote $H=(sa+sata)^*u$ the set of histories in~$M$ reaching~$u$.
	Observe also that the probabilities of histories~$H$ do not depend on the
	strategy. Let~$P_i(w)$ denote the probability of history~$w \in H$ in~$M_i$.
	We define $\sigma_\infty$ for any history~$w$ with~$w \in H$ as~$a$
	if~$P_1(w)\geq P_2(w)$ and as~$b$ otherwise.

	We first show that $\sum_{i=1,2} \pr_{M_i,s}^{\sigma_\infty}[\phi]
	= \sup_{\sigma}
	\sum_{i=1,2}\pr_{M_i,s}^{\sigma}[\phi]$, where~$\phi = \objreach(T)$, which
	proves that $\sigma_\infty$ achieves a Pareto-optimal probability vector.
	In fact, we have for any~$\sigma$ that
	$\pr_{M_i,s}^\sigma[\phi] = \sum_{w \in H} \pr_{M_i,s}^\sigma[\phi \mid
	w]P_i[w]$.
	So we get 
	$\pr_{M_1,s}^\sigma[\phi] = \linebreak \sum_{w \in H \cap \sigma^{-1}(\{a\})} P_1[w]$
	and $\pr_{M_2,s}^\sigma[\phi] = \sum_{w \in H \cap \sigma^{-1}(\{b\})} P_2[w]$.
	Since $H\cap \sigma^{-1}(\{a\})$ and $H\cap \sigma^{-1}(\{b\})$ paritions~$H$,
	we get that 
	$\sum_{i=1,2} \pr_{M_i,s}^\sigma[\phi] =
	\sum_{w \in H} P_{f(w)}(w)$, where $f(w)=1$ if~$w \in \sigma^{-1}(\{a\})$
	and~$2$ otherwise.
	On the other hand, by definition of~$\sigma_\infty$, we have
	$\sum_{i=1,2}\pr_{M_i,s}^{\sigma_\infty} \max(P_1(w),P_2(w))$.
	Since $Q(w) \leq \max(P_1(w),P_2(w))$ it follows that
	$\sum_{i=1,2}\pr_{M_i,s}^{\sigma_\infty}\geq 
	\sum_{i=1,2}\pr_{M_i,s}^{\sigma}$.

	Let us now show that no finite-memory strategy achieves 
	$\sup_{\sigma} \sum_{i=1,2}\pr_{M_i,s}^{\sigma}[\phi]$.
	Consider any $m$-memory strategy~$\sigma$ for arbitrary~$m>0$.
	Assume w.l.o.g. that $P_1(sas) > P_2(sas)$. Fix~$n=m^3$. Since $\sigma$ is
	finite-memory, there exists $0 \leq k_1<k_2 < m$ such that
	$\sigma$ has the same memory element after reading words $(sa)^n
	(sata)^{k_1}u$ and~$(sa)^n(sata)^{k_2}u$.  Let us write $w_{n,k} = (sa)^n (sata)^{k_1}u$.
	We have $\sigma( w_{n,k_1})=  \sigma( w_{n,k_2}) = \alpha \in
	\{a,b\}$. If~$\alpha=b$, then define $\sigma'$ identically as~$\sigma$ except
	for $\sigma'(w_{n,k_1}) = a$. We have $P_1(w_{n,k_1})>P_2(w_{n,k_2})$ so by
	the above calculations, $\sigma'$ achieves a higher objective than~$\sigma$.
	Assume that $\alpha=a$. In this case, we consider~$l$ large enough such that
	$P_2( w_{n,k_1+l(k_2-k_1)})>P_1(w_{n,k_1+l(k_2-k_1)})$. This holds for all
	large enough~$l$ since $P_2(sat)>P_1(sat)$. Moreover, on any word
	$\sigma(w_{n,k_1+l(k_2-k_1)})=a$ by the above pumping argument. 
	If we define $\sigma'$ by switching to~$b$ at this history, we again improve
	the objective function, similarly as above.
\end{proof}

\begin{figure}[h]
	\centering
	\begin{subfigure}[b]{0.2\textwidth}
		\centering
		\begin{tikzpicture}
			\tikzstyle{every state}=[node distance=1cm,minimum size=10pt, inner sep=1pt];
			\node[state] at (0,0) (A){$s$};
			\node[state,node distance=0.5cm, above right of=A,fill=black,minimum size=3pt] (Aa) {};
			\node[state,node distance=0.5cm, below right of=A,fill=black,minimum size=3pt] (Ab) {};
			\node at($(Aa)+(0,0.2)$) {$a$};
			\node at($(Ab)+(0,-0.2)$) {$b$};
			\node[state,right of=Aa] (B){$t$};
			\node[state,right of=Ab] (C){$u$};
			\node[node distance=0.5cm,below of=C]{T};
			\path[-latex']
				(A) edge (Aa)
				(A) edge (Ab)
				(Aa) edge[dotted] (B)
				(Aa) edge[dashed] (C)
				(Ab) edge[dashed] (B)
				(Ab) edge[dotted] (C)
				(B) edge[loop right] (B)
				(C) edge[loop right] (C);
		\end{tikzpicture}
		\caption{}
		\label{fig:randomization-required}
	\end{subfigure}
	\begin{subfigure}[b]{0.39\textwidth}
		\centering
		\begin{tikzpicture}
			\tikzstyle{every state}=[node distance=1cm,minimum size=10pt, inner sep=1pt];
			\node[state] at (0,0) (A){$s$};
			\node[state,right of=A] (B){$t$};
			\node[state,left of=A, node distance=0.6cm] (C){$u$};
			\node[state,node distance=0.5cm, above left of=C,fill=black,minimum size=3pt] (Ca) {};
			\node[state,node distance=0.5cm, below left of=C,fill=black,minimum size=3pt] (Cb) {};
			\node at($(Ca)+(0,0.2)$) {$a$};
			\node at($(Cb)+(0,-0.2)$) {$b$};
			\node[state,above of=A,node distance=0.8cm,fill=black,minimum size=3pt] (Ab) {};
			\node[node distance=0.2cm,above of=Ab]{$a$};
			\node[state,node distance=1cm, left of=Ca] (C1){$v$};
			\node[state,node distance=1cm, left of=Cb] (C2){$w$};
			\node[node distance=0.5cm,below of=C2]{T};
			\path[-latex']
				(B) edge node[above]{$a$} (A)
				(A) edge (Ab)
				(Ab) edge (C)
				(Ab) edge[bend left] (A)
				(Ab) edge[bend left] (B)
				(C) edge (Ca)
				(C) edge (Cb)
				(Ca) edge[dotted] (C1)
				(Ca) edge[dashed] (C2)
				(Cb) edge[dotted] (C2)
				(Cb) edge[dashed] (C1)
				(C1) edge[loop left] (C1)
				(C2) edge[loop left] (C2);
		\end{tikzpicture}
		\caption{}
		\label{fig:inf-memory-reach}
	\end{subfigure}
	\begin{subfigure}[b]{0.39\textwidth}
		\centering
		\begin{tikzpicture}
			\tikzstyle{every state}=[node distance=1cm,minimum size=10pt, inner sep=1pt];
			\node[state] at (0,0) (A){$s$}; 
			\node[state,right of=A, node distance=1.2cm] (B){$t$}; 
			\node[state,node distance=0.7cm,above right of=A,fill=black,minimum size=3pt] (Aa) {};
			\node[above of=Aa, node distance=0.2cm]{$a$};
			\node[state,node distance=0.7cm, below right of=A,fill=black,minimum size=3pt] (Bb) {};
			\node[above of=Bb, node distance=0.2cm]{$a$};
			\node[state,node distance=0.5cm, above left of=A,fill=black,minimum size=3pt] (Cc) {};
			\node[state,node distance=0.5cm, below left of=A,fill=black,minimum size=3pt] (Cb) {};
			\node at($(Cc)+(0,0.2)$) {$c$};
			\node at($(Cb)+(0,-0.2)$) {$b$};
			\node[state,left of=Cc] (C1){$u$};
			\node[state,left of=Cb] (C2){$v$};
			\node at ($(C2)+(0,-0.4)$) {$T$};
			\path[-latex']
				(A) edge (Aa)
				(Aa) edge[bend left] (A)
				(Aa) edge (B)
				(B) edge (Bb)
				(Bb) edge (A)
				(A) edge (Cb)
				(A) edge (Cc)
				(Cb) edge[dashed] (C1)
				(Cb) edge[dotted] (C2)
				(Cc) edge[dotted] (C1)
				(Cc) edge[dashed] (C2);
		\end{tikzpicture}
		\caption{}
		\label{fig:unachievable}
	\end{subfigure}
	\caption{
		\small
		We adopt the following notation in all examples: edges that only exist
		in~$M_1$ are drawn in dashed lines, and those that only exist in~$M_2$
		by dotted ones.
		To see that randomization may be necessary, observe that in the MEMDP~$M$ in Fig.~\ref{fig:randomization-required},
		the vector $(0.5,0.5)$ of reachability probabilities for target~$T$ can only be achieved by a strategy that
		randomizes between $a$ and~$b$.
		In the MEMDP in Fig.~\ref{fig:inf-memory-reach}, where action $a$
		from~$s$ has the same
		support in $M_1$ and~$M_2$ but different distributions.
		Any strategy almost surely reaches~$u$ in
		both~$M_i$, since action~$a$ from~$s$ has nonzero probability of leading
		to~$u$. Intuitively, the best strategy is to sample the
		distribution of action~$a$ from~$s$, and to choose, upon arrival to~$u$,
		either~$b$ or~$c$ according to the most probable environment.
		We prove that such an infinite-memory strategy achieves a Pareto-optimal vector
		which cannot be achieved by any finite-memory strategy (See Lemma~\ref{lemma:inf-memory-reach} in Appendix).
		Last, in Fig.~\ref{fig:unachievable}, the MEMDP is similar to that of
		Fig.~\ref{fig:inf-memory-reach} except that 
		action~$a$ from~$s$ only leads to~$s$ or~$t$.
		We will prove in Section~\ref{section:limitsure}, that for
		any~$\epsilon>0$, there exists a strategy ensuring
		reaching~$T$ with probability $1-\epsilon$ in each~$M_i$. The strategy
		consists in sampling the distribution of action~$a$ from~$s$ a sufficient number of times
		and estimating the actual environment against which the controller is
		playing.
		However, the vector~$(1,1)$ is not achievable, which follows from
		Section~\ref{subsection:almostsure-reach-safe}.
	}
	\label{fig:counterexamples}
\end{figure}
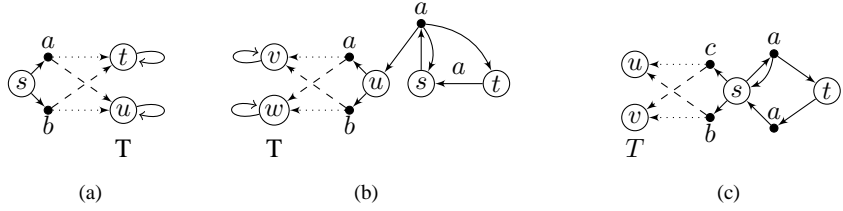

\textit{Results}~
We give efficient algorithms for almost-sure and limit-sure problems:

%
\smallskip
	(A) The almost-sure reachability, safety, and parity problems are decidable in
	polynomial time (Theorems~\ref{lemma:asreach} and~\ref{thm:asparity}).
	Finite-memory strategies suffice.

(B)
	The limit-sure reachability, safety, and parity problems are decidable in
	polynomial time (Theorem~\ref{thm:lsreach} and~\ref{thm:quantparity}). Moreover, for any $\epsilon>0$, to achieve probabilities of
	at least $1-\epsilon$, 
	$O(\frac{1}{\eta^2}\log(\frac{1}{\epsilon}))$-memory strategies suffice,
	where~$\eta$ denotes the smallest positive difference between the probabilities
	of~$M_1$ and~$M_2$.
\smallskip

The general quantitative problem is harder as shown by the next result.
We call a MEMDP \emph{acyclic} if the only cycles are self-loops in all
environments.

\smallskip
(C)
	The quantitative reachability and safety problems are NP-hard on acyclic
	MEMDPs both for arbitrary and memoryless strategies
(Theorem~\ref{thm:nphard}).
\smallskip

We can nevertheless provide procedures to solve the quantitative reachability
and safety problems by fixing the memory size of the strategies.

\smallskip
(D)
	For any~$K$ represented in unary, the quantitative reachability and safety
	problems restricted to $K$-memory strategies can be solved in PSPACE
	(Theorem~\ref{lemma:fixed-memory-reach}).
\smallskip


The quantitative parity problem can be reduced to quantitative
reachability, so the previous result can also be applied for the quantitative
parity problem.

\smallskip (E)
	The quantitative parity problem can be reduced to quantitative
	reachability in polynomial time
	(Theorem~\ref{thm:quantparity}).
	\smallskip

We show that finite-memory strategies are not restrictive if we are interested
in approximately ensuring given probabilities.

\smallskip
(F)
	Finite-memory strategies suffice to approximate quantitative
	reachability, safety, and parity problems up to any desired precision
	(Theorem~\ref{lemma:reach-finite-suffices}).
\smallskip

We will derive approximation algorithms in the
following sense.
\begin{definition}
	The \emph{$\epsilon$-gap problem for reachability} consists, given MEMDP~$M$, state~$s$, target set~$T$, and
	probabilities~$\alpha_1,\alpha_2$, in answering

		-- YES if~$\exists \sigma, \forall i=1,2,
			\pr_{M_i,s}^{\sigma}[\objreach(T)]\geq \alpha_i$,

		--  NO if~$\forall \sigma, \exists i=1,2, \pr_{M_i,s}^{\sigma}[\objreach(T)]<
			\alpha_i - \epsilon$,

		--  and arbitrarily otherwise.
\end{definition}
The $\epsilon$-gap problem is an instance of promise problems
which  guarantee a correct answer in two disjoint sets of
inputs, namely positive and negative instances -- which do not necessarily cover all
inputs, while giving no guarantees in the rest
of the input~\cite{ESY-ic84,goldreich2005promise}.


We give a procedure for the~$\epsilon$-gap problem and show its NP-hardness:

\smallskip (G)
	There is a procedure for the~$\epsilon$-gap problem for quantitative
	reachability in MEMDPs that runs in double exponential space, and
	whenever it answers YES, returns a strategy~$\sigma$
	such that $\pr_{M,s}^{\sigma}[\objreach(T)]\geq \alpha_i-\epsilon$
	(Theorem~\ref{lemma:gapalgorithm}).


(H) The $\epsilon$-gap problem is NP-hard (Theorem~\ref{lemma:gaphard}).
\smallskip

\paragraph{Preprocessing}
\label{section:preprocessing}
Clearly, in a MEMDP, if one observes an edge that only exists in one
environment, then the environment is known with certainty
and any good strategy should immediately switch to the optimal strategy for the
revealed environment. 
Formally, we say that an edge $(s,a,s')$ is \emph{$i$-revealing} if 
$\delta_i(s,a,s') \neq 0$ and $\delta_{3-i}(s,a,s') = 0$.
We make the following assumption w.l.o.g.:

\begin{assumption}[Revealed form]
	\label{assum:revealed}
	All MEMDPs $M=(S,A,\delta_1,\delta_2)$ are assumed to be in \emph{revealed
		form}, that is, 
	there exists a partition $S = S_u \biguplus R_1
	\biguplus R_2$ satisfying the following properties.
	\begin{inparaenum}
		\item All states of~$R_1$ and~$R_2$ are absorbing in both environments,
		\item For any~$i=1,2$, and any $i$-revealing edge $(s,a,s')$, we have $s'
		\in R_i$. Conversely, any edge $(s,a,s')$ with~$s' \in R_i$ is
		$i$-revealing.
	\end{inparaenum}

	States $R_i$ are called \emph{$i$-revealed}, and will be denoted $R_i(M)$.
	The remaining states are called \emph{unrevealed}.
\end{assumption}
In other words, we assume that any $i$-revealing edge leads to a known
set of $i$-revealed states which are all absorbing.
Assumption~\ref{assum:revealed} can be made without loss of generality by
redirecting any revealing edge to fresh absorbing states.
In fact, given an arbitrary MEMDP~$M$, for any objective~$\Phi$,
we can define $M'$ by replacing any $i$-revealing edge $(s,a,s')$ in~$M$ by two edges
$(s,a,\top_i)$ and~$(s,a,\bot_i)$ where
$\top_i$ (resp. $\bot_i$) is a fresh absorbing winning (resp. losing) state.
Here, by winning, we mean that we add~$\top_i$ (resp. $\bot_i$) to the set of
target (resp. non-target) states for
reachability objectives, to the set of safe (resp. unsafe) states for safety objectives, and we
assign an even (resp. odd) parity for parity objectives. The probabilities are
defined as follows: $\delta_i'(s,a,\top_i)= \delta_i(s,a,s')\cdot 
\Val_{\Phi}^*(M_i,s')$, and
$\delta_i'(s,a,\bot_i)= \delta_i(s,a,s') \cdot (1-\Val^*_{\Phi}(M_i,s'))$, while the
probabilities of other edges are preserved. The interpretation of these values
is that at state~$s$, given action~$a$, $\delta_i(s,a,s')\cdot
\Val_{\Phi}^*(M_i,s')$ is the probability of going to~$s'$, and from thereon
winning under the optimal strategy for~$M_i$. The construction is illustrated
in Fig.~\ref{fig:revealed}.

Note that from any strategy~$\sigma'$ in~$M'$
one can derive, by adding one bit of memory, a strategy~$\sigma$ for~$M$ such
that
$\pr_{M_i,s_0}^\sigma[\Phi] = \pr_{M_i',s_0}^{\sigma'}[\Phi],
\forall i=1,2$,
and 
$\expect_{M_i,s_0}^\sigma[\Phi] = \expect_{M_i',s_0}^{\sigma'}[\Phi],
\forall i=1,2$
respectively for considered objectives.
Similarly, any strategy in~$M$ can be adapted to~$M'$ preserving the
probabilities of satisfying a given objective.

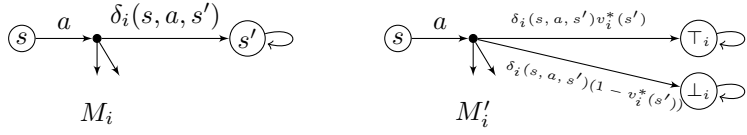
\begin{figure}[h]
	\centering
	\begin{tikzpicture}
		\tikzstyle{every state}=[node distance=1cm,minimum size=10pt, inner sep=1pt];
		\node[state] at (0,0) (A){$s$};
		\node[state,right of=A,fill=black,minimum size=3pt] (Aa) {};
		\node[state, right of=A, node distance=3cm] at (0,0) (B){$s'$};
		\path[-latex'] (A) edge node[above]{$a$} (Aa)
			(Aa) edge node[above]{$\delta_i(s,a,s')$} (B)
			(Aa) edge ($(Aa) + (0,-0.5)$)
			(Aa) edge ($(Aa) + (0.3,-0.5)$)
			(B) edge[loop right] (B);
		\node at ($(Aa)+(0,-1)$) {$M_i$};
		\begin{scope}[shift={(5,0)}]
			\node[state] at (0,0) (A){$s$};
			\node[state,right of=A,fill=black,minimum size=3pt] (Aa) {};
			\node[state, right of=A, node distance=4cm] (B){\scriptsize $\top_i$};
			\node[state, below of=B, node distance=0.7cm] (C){\scriptsize $\bot_i$};
			\path[-latex'] (A) edge node[above]{$a$} (Aa)
				(Aa) edge node[sloped,pos=0.6,below]{\tiny
					$\delta_i(s,a,s')(1-v_{i}^*(s'))$} (C)
				(Aa) edge node[above]{\tiny $\delta_i(s,a,s') v_{i}^*(s')$} (B)
				(Aa) edge ($(Aa) + (0,-0.5)$)
				(Aa) edge ($(Aa) + (0.3,-0.5)$)
				(B) edge[loop right] (B)
				(C) edge[loop right] (C);
			\node at ($(Aa)+(0,-1)$) {$M_{i}'$};
		\end{scope}
	\end{tikzpicture}
	\caption{The transformation of any $i$-revealing edge $(s,a,s')$ so as to
		put the MEMDP in revealed form, where 
	$v_{i}^*(s') = \Val^*_{\Phi}(M_{i},s')$, for considered
	objective~$\Phi$.}
	\label{fig:revealed}
\end{figure}




For any reachability (resp. safety) objective~$T$, once a state in~$T$ (resp.
$S\setminus T$) is visited the behavior
of the strategy afterwards is not significant since the objective has already
been fulfilled (resp. violated).
Accordingly, we assume that the set of target and unsafe states are absorbing.

\begin{assumption}
	\label{assum:absorbing}
	For all considered objectives~$\objreach(T)$ and $\Wsafety(T')$,
	we assume that $T$ and~$S\setminus T'$ are sets of absorbing states for both environments.
\end{assumption}

Under assumptions \ref{assum:revealed} and~\ref{assum:absorbing}, for any
MEMDP~$M$, and objective~$\Phi$, 
we denote
$R_i^{\Phi}(M)$ the set of $i$-revealed states from which
$\Phi$ holds almost surely in~$M_i$, and define $R^{\Phi}(M) = R_1^{\Phi}(M)
	\cup R_2^{\Phi}(M)$.

\paragraph{Overview}~
We will first concentrate on results on reachability objectives since they
contain most of the important ideas. We present algorithms for almost-sure
reachability (Section~\ref{section:reachsafe}), introduce and study \emph{double
end-components} (Section~\ref{section:dec}), then present our
algorithms for limit-sure problems (Section~\ref{section:limitsure}),
and the general quantitative case where we also present NP-hardness results
(Section~\ref{section:quantitative}). We then summarize our results on safety,
and parity objectives (Section~~\ref{section:parity}).

\section{Almost-Sure Reachability}
\label{section:reachsafe}
\label{subsection:almostsure-reach-safe}


We give polynomial-time algorithms for almost-sure reachability in
MEMDPs.
Given any MEMDP~$M=(S,A,\delta_1,\delta_2)$, we define the MDP $\cup M =
(S,A,\delta)$ by taking, for each action, the union of all transitions, and
assigning them uniform probabilities. Formally, for any~$s\in S$ and $a \in A(s)$, 
$\supp(\delta(s,a)) = \supp(\delta_1(s,a)) \cup \supp(\delta_2(s,a))$
and for any $s' \in \supp(\delta(s,a))$, $\delta(s,a,s') =
\frac{1}{|\supp(\delta(s,a))|}$.

Observe that for any MEMDP~$M$, and subset of states~$S'$, the set of states $s$ such that
$\pr_{\cup M,s}^\sigma[\Wsafety(S')]=1$ for some~$\sigma$
induces a sub-MDP in $M_1$ and~$M_2$. One can
therefore define~$M'$ the MEMDP induced by this set. Furthermore, any strategy
compatible with~$M'$ satisfies $\Wsafety(S')$ surely in each~$M_i$.

\begin{figure}[h]
	\centering
	\begin{tikzpicture}
		\tikzstyle{every state}=[node distance=1cm,minimum size=10pt, inner sep=1pt];
		\node[state] at (0,0) (A){$s$};
		\node[state,right of=A,fill=black,minimum size=3pt] (Aa) {};
		\node[state,below of=Aa,node distance=0.5cm, fill=black,minimum size=3pt] (Ab) {};
		\node[state,right of=Aa] (B){$t$};
		\node[state,right of=Ab] (C){$u$};
		\node[right of=B, node distance=0.5cm] {$T$};
		\node[right of=C, node distance=0.5cm] {$T$};
		\node at ($(Aa) +(0,0.3)$) {$a$};
		\node at ($(Ab) +(0,-0.3)$) {$b$};
		\path[-latex'] 
			(A) edge (Aa)
			(A) edge (Ab)
			(Aa) edge[dashed, bend right] (A)
			(Aa) edge[dotted] (B)
			(Ab) edge[dashed] (C)
			(Ab) edge[dotted,bend left] (A)
			(B) edge[loop above] (B)
			(C) edge[loop below] (C);
	\end{tikzpicture}
	\caption{\small MEMDP~$M$ where $\objreach(T)$ can be achieved almost surely.
	In fact, $\almostsure(M_i,T)=\{s,t,u\}$ for all~$i=1,2$, so~$M'=M$, and
	$\Val_{\objreach(T)}(M_i',s)=1$ for~$i=1,2$. The strategy returned by the
	algorithm consists in choosing, at~$s$, $a$ and~$b$ uniformly at random.
	Notice that there is no pure memoryless strategy achieving the objective
	almost surely.}
	\label{fig:aspositive}
\end{figure}
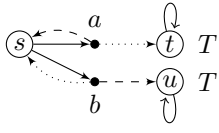

The algorithm for almost sure reachability is described in Algorithm~\ref{alg:asreach}.
First, the state space is restricted
to~$U$ since any state from which the objective holds almost surely in the
MEMDP~$M$ must also belong to an almost surely winning state of each $M_i$,
except for $j$-revealed states which only need to be winning for~$M_j$.
We consider MEMDP~$M'$ induced by the states surely satisfying~$\Wsafety(U)$
in both environments.
The problem is then reduced to finding strategies in each~$M_i'$. 
If such strategies we obtain our strategy either 1) alternating between two
strategies using memory, or 2) randomizing between them.
Figure~\ref{fig:aspositive} is an example where almost-sure reachability holds.
We already saw the example of Fig.~\ref{fig:unachievable} where almost-sure
reachability does not hold. In fact, in that example $M'$ contains both states
$\{s,t\}$ both no winning strategy exists in $M_i'$ for both~$i=1,2$.
\renewcommand{\KwData}[1]{\textbf{Input:} #1\\}
\begin{algorithm}
 \KwData{MEMDP~$M$, $\objreach(T)$, $s_0 \in S$}
 $U := \big(\almostsure(M_1,\objreach(T)) \cap
 \almostsure(M_2,\objreach(T))\big)
			 \cup R^{\objreach(T)}$\;
 $M'$ := Sub-MEMDP of~$M$ induced by states~$s$ s.t.  $\Val_{\Wsafety(U)}^*(\cup M,s)=1$\;
 \eIf{$\forall i=1,2, \Val_{\objreach(T)}^*(M_i',s_0)=1$}{
	 Let $\sigma_i$ for $i=1,2$, such that $\Val_{\objreach(T)}^\sigma(M_i',t)=1$ for all~$t \in U$\;
	 Return $\sigma'$ defined as $\sigma'(t) = \frac{1}{2}\sigma_1(t)
		 +\frac{1}{2}\sigma_2(t)$, $\forall t\in S$\;
 }{
 	Return NO\;
 }
 \caption{\small Almost-sure reachability algorithm given MEMDP~$M$, starting
	 state~$s_0$
 and objective $\objreach(T)$.}
	\label{alg:asreach}
\end{algorithm}


\begin{theorem}
	\label{lemma:asreach}
	For any MEMDP~$M$, objective~$\objreach(T)$,
	and a state~$s$, Algorithm~\ref{alg:asreach} decides in polynomial time if
	$\objreach(T)$ can be achieved almost surely from~$s$
	in~$M$, and returns a witnessing memoryless strategy.
\end{theorem}
\begin{proof}
	\textbf{(Soundness)}
	Assume that $\forall i=1,2, \Val_{ \objreach(T)}^*(M_i',s)=1$,
	and consider pure
	memoryless strategies $\sigma_i$ achieving $\objreach(T)$
	almost surely in each $M_i'$ from any state of~$U$, 
	and let $\sigma = \frac{1}{2}\sigma_1 + \frac{1}{2}\sigma_2$.
	We have $\pr_{M_i,s}^\sigma[\Wsafety(U)]=1$ for any~$i$ since
	each~$\sigma_i$ is compatible with~$M_i'$.
	Moreover, for each~$i$, and from any state~$s'$ of~$M_i$ reachable
	under~$\sigma$, target set $T$ is reached with
	positive probability in $|S|$ steps under strategy $\sigma_i$.
	In fact we have, for such a state~$s'$, $s' \in U \setminus R_{3-i}$.
	Since the probability
	of~$\sigma$ being identical to~$\sigma_i$ for~$|S|$ steps is positive, 
	$T$ is reached almost surely in $M_i$ under~$\sigma$
	from~$s$. 

	This construction gives a memoryless strategy. One can obtain a pure
	finite-memory strategy by alternating between $\sigma_1$ and $\sigma_2$ every
	$|S|$ steps.

	\textbf{(Completeness)}
	Conversely, assume that there exists a strategy $\sigma$ almost surely
	achieving $\objreach(T)$ from~$s$. 
 Towards a contradiction, assume that $\Val_{\Wsafety(U)}^\sigma(M_i,s)< 1$.
 This means some state~$t \not \in U$ is reached with positive probability
 under~$\sigma$.
 Recall that all target states are absorbing in~$M$ by
 Assumption~\ref{assum:absorbing}.
 If~$t \in R_i\setminus R_i^{\objreach(T)}$ this contradicts that~$\sigma$
 almost surely achieves the objectives, and similarly if~$t \not \in
 \almostsure(M_i,T)$, since a target state could not have been reached
 before arriving to~$t$. Last, if $t \not \in \almostsure(M_{3-i},T)$
 and $t$ is not revealed, then this state is also reachable with positive
 probability in~$M_{3-i}$ under~$\sigma$, which is again a contradiction.
 Therefore $\Val^\sigma_{\Wsafety(U)}(M_i,s)=1$ for all~$i=1,2$, which means
 that $s$ is a state of~$M'$ and~$\sigma$ is compatible with~$M'$.
 Last, we do have $\Val^*_{\objreach(T)}(M_i',s_0)=1$ since~$\sigma$ is a
 witnessing strategy.
 Therefore, the algorithm answers positively on this instance.
\end{proof}

\section{Double end-components}
\label{section:dec}
End-components play an important role in the analysis of
MDPs~\cite{DeAlfaro-phd97}.
Because the probability distributions in different environments of an MEMDP
can have different supports, we need to adapt the notion for MEMDPs.
We thus introduce \emph{double end-components} which are sub-MDPs that
are end-components in both environments. 
We show that one can \emph{learn} inside double end-components, and use
these observations to study limit-sure objectives.

Formally, 
	given a MEMDP $M=(S,A,\delta_1,\delta_2,r)$, 
	a \emph{double end-component (DEC)} 
	is a pair $(S',A')$ where~$S' \subseteq S$, and $A' \subseteq A$
	such that $(S',A')$ is an end-component in each~$M_i$.
A double end-component $(S',A')$ is \emph{distinguishing} if there exists $(s,a)
\in S' \times A'$ such that $\delta_1(s,a) \neq \delta_2(s,a)$.
The union of two DECs with a common state is a DEC; we
consider maximal DECs (MDEC).
MDECs can be computed in polynomial time by first eliminating from $M$ all actions with
different supports, and then computing the MECs in the remaining MDPs.
A DEC is \emph{trivial} if it is an absorbing state.

Under Assumption~\ref{assum:absorbing}, 
for reachability objectives, a DEC is \emph{winning} if it is an
absorbing state winning for the objective.
A DEC~$D$ is \emph{winning} for a parity objective~$\Phi$, 
if there exists a strategy compatible with~$D$ satisfying $\Phi$ almost surely;
Lemma~\ref{lemma:mdp-mec-parity} shows that a common strategy exists for both
environments.

We first solve the problems of interest in distinguishing DECs
up to any error bound~$\epsilon$. 
The idea is that in a distinguishing DEC, one can learn the environment by
sampling the distribution of distinguishing actions.
\begin{lemma}
	\label{lemma:strategy-for-mec-as}
	Consider any MEMDP~$M=(S,s_0,A,\delta_1,\delta_2)$, a distinguishing double
	end-component $D = (S',A')$, state $s \in S'$,~$\epsilon>0$, and any 
	objective $\Phi$ reachability, safety, parity.
	For any~$\epsilon>0$, there exists a strategy~$\sigma$ such that
	\(
		\pr_{M_i,s}^\sigma[\Phi] \geq (1-\epsilon)\Val_{\Phi}^*(M_i,s), \forall i=1,2.
	\)
\end{lemma}
\begin{proof}

	Fix $(s,a) \in S' \times A'$ such that $\delta_1(s,a) \neq \delta_2(s,a)$.
	The strategy runs in two rounds. In the first round, the
	goal is to sample the distribution of the edge $(s,a)$. For this, it suffices
	to execute a strategy that chooses each available action compatible with~$D$ uniformly at random, 
	and upon arrival to state~$s$, to choose action~$a$, and store the number of times the next
	state is~$s'$.
	After~$K$ visits to~$s$, we make a guess about the current MDP
	depending on the sampled value. The second round of
	the strategy is the memoryless optimal strategy in one of the~$M_i$.
	When $K$ is chosen sufficiently large, we obtain the desired result.
	
	Let us denote $d_i=\delta_i(s,a,s')$ for some~$s'$ satisfying $d_1\neq d_2$,
	and assume w.l.o.g. that $d_1 < d_2$.
	For any~$\epsilon>0$, let~$K = 2\frac{\log(1/\epsilon)}{(d_2-d_1)^2}$,
	and let~$f$ be a memoryless strategy which chooses uniformly at
	random all actions except action~$a$ is picked at~$s$ deterministically.
	Under~$f$, each state is visited infinitely often almost surely.
	We define $f_K$ by augmenting~$f$ with memory as follows. Informally, $f_K$ has two counters:
	$c_{s,a}$ counting the number of visits at~$s$, and $c_{s,a,s'}$ counting the
	occurence of edge $(s,a,s')$. Hence, at each visit at~$s$, we have a
	Bernouilli trial with mean $\delta_i(s,a,s')$ (for each~$i$), and~$c_{s,a,s'}$ is the
	number of successful trials.
	It is clear that the ratio $c_{s,a,s'}/c_{s,a}$ should go to $\delta_i(s,a,s')$ inside
	each~$M_i$. We execute this strategy until $c_{s,a}=K$, which happens almost
	surely.
	We complete the description of strategy~$f_K$ by extending it,
	once $c_{s,a,s'}=K$ is reached,
	with the optimal memoryless strategy $\text{opt}_1$ for~$M_1$ if
	$\frac{c_{s,a,s'}}{c_{s,a}} \leq \frac{d_1+d_2}{2}$,
	and $\text{opt}_2$, the one for $M_{2}$ otherwise.

	By Hoeffding's inequality, we have
	\[
		\pr_{M_1,s}^{f_K}[\frac{c_{s,a,s'}}{c_{s,a}} \geq d_1 +
		\frac{d_2-d_1}{2}\mid c_{s,a} =K]
		\leq e^{-2K\frac{d_2-d_1}{2}^2} \leq \epsilon.
	\]
	and
	\[
		\pr_{M_2,s}^{f_K}[\frac{c_{s,a,s'}}{c_{s,a}} \leq d_2  - \frac{d_2-d_1}{2}\mid c_{s,a} =K]
		\leq e^{-2K\frac{d_2-d_1}{2}^2} \leq \epsilon.
	\]

	We now compute the values under strategy~$f_K$, distinguishing whether the sampled frequency stays
	within the given radius or not. In the first case, the objective is satisfied
	with probability $\Val_\Phi^*(D)$,
	and in the second case, with probability at least~$0$.
	It follows that $\pr_{M_i,s}^{f_K}[\Phi] \geq (1-\epsilon)\Val_{\Phi}^*(D)$.
	Note that the memory requirement is $K^2$, since we store the pairs $(c_{s,a},
	c_{s,a,s'})$.
\end{proof}
\begin{remark}
	The algorithm can be improved in practice as follows. Let~$S'$ denote the set
	of states of the end-component which have distinguishing actions. For any
	state~$s \in S'$, fix a distinguishing action $a_{s}$. For any $s'$ such
	that $\delta_1(s,a_s,s')\neq \delta_2(s,a_s,s')$, write 
	$K_{s,a_s,s'}$ the above constant computed for this edge.
	We apply the following strategy: at any state~$s \in S'$ play~$a_{s}$,
	and sample the distribution. At any state~$s \not \in S'$, pick an action uniformly
	at random.
	Now, we run this strategy until we collected $K_{s,a_s,s'}$ samples for some
	action~$a_s$.
	Note that if $S'$ is a singleton, this does not improve the lemma's proof.

	What expected time can we guarantee until the environment is guessed with
	prob. $1-\epsilon$? Let $T(s',s)$ denote the expected time to reach state~$s$
	from~$s'$ under the uniform strategy
	\footnote{Note that since we do not know the exact
		distributions, we cannot minimize the expected time 
			using an optimal strategy here.}, and let $T(s) = \max_{s'} T(s',s)$.
	If~$s$ denotes a state with a distinguishing action, such that 
	$\eta = |\delta_1(s,a,s') - \delta_2(s,a,s')|$, then the above algorithm
	switches to a pure optimal strategy in expected
	$O(T(s)\frac{\log(1/\epsilon)}{\eta^2})$ time.
\end{remark}

We now consider general MEMDPs, and define a transformation by contracting
DECs. The transformation preserves, up to any desired~$\epsilon$, the
probabilities of objectives, thanks to Lemma~\ref{lemma:strategy-for-mec-as}.

Given a DEC $D=(S',A')$, a \emph{frontier state}~$s$ of~$D$
is such that there exist~$a \in A(s) \setminus A'(s)$, $i \in \{1,2\}$,
and $s' \not \in S'$ such that $\delta_i(s,a,s') \neq 0$.
An action~$a \in A(s) \setminus A'(s)$ is a \emph{frontier action} for~$D$.
A pair $(s,a)$ is called \emph{frontier state-action} when~$a \in A(s)$ is
a frontier action.

\begin{definition}
	\label{def:tilde}
	Given a MEMDP $M = (S,A,\delta_1,\delta_2)$, and reachability or safety objective
	$\Phi$, we define 
	$\hat{M}=(\hat{S},\hat{A},\hat{\delta_1},\hat{\delta_2})$ as
	follows.
	\begin{inparaenum}[a)]
		\item 
		Any distinguishing MDEC~$D$ is contracted as in
		Fig.~\ref{fig:winning-dec}
		where in~$M_i$, action~$a$ leads to new states~$W_D$
		with probability $v_i=\Val_{\Phi}^*(M_i,D)$, and to $L_D$ with 
		probability $1-v_i$.
		\item 
		Any non-distinguishing MDEC~$D = (S',A')$ is replaced with the
		module in Fig.~\ref{fig:losing-dec}.
		The actions $a_D^\$$ and $\{f_ia_i\}_{(f_i,a_i) \in F}$ are available 
		from~$s_D$ where $F$ is the set of pairs of frontier state-actions of~$D$.
		For any~$(f_i,a_i)$,
		the distribution $\hat{\delta}_j(s_D,f_ia_i)$ is obtained from
		$\delta_j(f_i,a_i)$ by redirecting to~$s_D$ all edges that
		lead inside~$S'$. 
	\end{inparaenum}

	We define the new objective $\hat{\Phi}$ by restricting~$\Phi$
	to~$\hat{S}$, and adding all states~$W_D$ in the target (resp. safe) set.
\end{definition}

We denote by~$\hatabstr: S\rightarrow \hat{S}$ the mapping from the states
of~$S$ to that of~$\hat{S}$ defined by the above transformation, mapping
any state~$s$ of a DEC~$D$ is to~$s_D$, and any other state to itself.
We will also denote $\hat{s} = \hatabstr(s)$.

\begin{figure}[h]
	\centering
	\begin{subfigure}[b]{0.4\textwidth}
	\begin{tikzpicture}
		\tikzstyle{every state}=[node distance=1cm,minimum size=10pt, inner sep=1pt];
		\node[state] at (0,0) (A){$s_D$};
		\node[state,right of=A,fill=black,minimum size=3pt] (Aa) {};
		\node[state] at ($(Aa)+(1,0.5)$) (B){$W_D$};
		\node[state] at ($(Aa)+(1,-0.5)$) (C){$L_D$};
		\path[-latex']
			(A) edge node[above]{$a_D^\$$} (Aa)
			(Aa) edge node[above,sloped] {\small $v_i$}(B)
			(Aa) edge node [below,sloped] {\scriptsize $1-v_i$}(C)
			(B) edge[loop right] (B)
			(C) edge[loop right] (C);
	\end{tikzpicture}
	\caption{Reducing distinguishing DECs, where $v_i=\Val_{\phi_i}^*(M_i,D)$.}
	\label{fig:winning-dec}
\end{subfigure}
	\qquad
	\begin{subfigure}[b]{0.4\textwidth}
		\begin{tikzpicture}
			\tikzstyle{every state}=[node distance=1cm,minimum size=10pt, inner sep=1pt];
			\node[state] at (0,0) (A){$s_D$};
			\node[state,left of=A,fill=black,minimum size=3pt] (Aa) {};
			\node[state] at ($(Aa)+(-1.3,0.5)$) (B){$W_D$};
			\node[state] at ($(Aa)+(-1.3,-0.5)$) (C){$L_D$};
			\node[state,above right of=A,fill=black,minimum size=3pt] (F3){};
			\node[state,below of=A,fill=black,minimum size=3pt] (F1a) {};
			\node[state,below right of=A,fill=black,minimum size=3pt] (F2a) {};
			\path[-latex']
				(A) edge node[above]{$a_D^\$$} (Aa)
				(A) edge node[right]{$s_m a_m$}(F3)
				(A) edge node[below left]{$s_1a_1$} (F1a)
				(A) edge node[above right]{$s_2a_2$} (F2a)
				(F1a) edge[gray,bend right] (A)
				(F2a) edge[gray,bend right] (A)
				(F1a) edge[gray] ($(F1a)+(-0.5,-0.1)$)
				(F2a) edge[gray] ($(F2a)+(1,0)$)
				(F2a) edge[gray] ($(F2a)+(0.6,-0.2)$)
				(F3) edge[gray] ($(F3)+(0.5,0)$)
			(Aa) edge node[above,sloped] {\small $\restr{v_i}{D}$}(B)
			(Aa) edge node [below,sloped] {\scriptsize $1-\restr{v_i}{D}$}(C)
			(B) edge[loop left] (B)
			(C) edge[loop left] (C);
		\end{tikzpicture}
		\caption{Reduction of non-distinguishing DECs,
		where $\restr{v_i}{D} = \restr{\Val_{\Phi}^*}{D}(M_i,D)$. }
		\label{fig:losing-dec}
	\end{subfigure}
\end{figure}

The intuition is that when the play enters a distinguishing DEC~$D$, by
applying Lemma~\ref{lemma:strategy-for-mec-as}, we can arbitrarily approximate
probabilities~$v_i = \Val_{\Phi}^*(M_i,D)$.
From a state~$s$ in a non-distinguishing component~$D$ in~$M$, the play either stays
forever inside and obtain the value $\restr{\Val_{\Phi}^*}{D}(M_1,s) =
\restr{\Val_{\Phi}^*}{D}(M_2,s)$ (as it is non-distinguishing), or it eventually leaves~$D$.
The first case is modeled by the
action~$a_D^\$$, and the second case by the remaining actions leading to frontier
states. Note that there is a strategy under which, from any state of~$D$, in
$M_1$ and~$M_2$, all states and actions of~$D$ are visited infinitely often
(by considering a memoryless strategy choosing all actions
uniformly at random -- see \textit{e.g.} \cite{Puterman-wiley94}). 
We will use this construction for reachability and safety objectives;
while a specialized construction based on~$\hat{M}$ will be defined for parity
objectives.

The point in defining~$\hat{M}$ is to eliminate all non-trivial DECs:
\begin{lemma}
	\label{lemma:mec-transient}
	Let~$D$ be a maximal end-component of~$\hat{M}_i$.
	Then either~$D$ is a trivial DEC, or~$D$ is transient in $\hat{M}_{3-i}$.
\end{lemma}
	\begin{proof}
		Assume that~$D$ is an end-component of~$\hat{M}_{3-i}$. Then~$D$ is a double
		end-component by definition. If~$D$ is a self-loop, then it is an absorbing
		state and we are done.
		Otherwise, $D$~must contain some state $s_E$ of~$\hat{M}$ created by 
		contracting MDEC~$E$
		since otherwise~$D$ would have been contracted itself by
		definition of~$\hat{M}$.
		But then $D\cup E$ is a DEC larger than~$D$, which is a contradiction.
		Thus, $D$ cannot be an end-component of~$\hat{M}_{3-i}$ unless it is one absorbing state.
		
		Assuming~$D$ is not an end-component of~$\hat{M}_{3-i}$, either~$D$ is not strongly connected,
		or it is not $\delta_{3-i}$-closed.
		Observe first that $D$ does not contain $i$-revealing edges in~$\hat{M_i}$ since otherwise, by
		construction of~$\hat{M_i}$, it contain an absorbing state and not be
		strongly connected in~$\hat{M_i}$.
		We show that~$D$ must also be strongly connected in~$\hat{M}_{3-i}$. In fact, assume otherwise and
		consider two states~$s$ and~$t$ such that~$t$ is not reachable from~$s$
		in~$\restr{\hat{M}_{3-i}}{D}$. Along the run from~$s$ to~$t$, $\hat{M_i}$
		must have an edge that is absent from~$\hat{M}_{3-i}$, which is an
		$i$-revealing edge;
		contradiction. Therefore, $D$ is strongly connected and not
		$\delta_{3-i}$-closed in~$\hat{M}_{3-i}$.

		We now show that under any strategy in~$\hat{M}_{3-i}$, the play eventually
		leaves~$D$
		almost surely. It suffices to show that $\hat{M}_{3-i}$
		has no end-component inside~$D$. Let~$D'\subseteq D$ be such an end-component.
		Then~$D'$ does not contain $3-i$-revealing edges; in fact, we know that~$D$ is
		strongly connected, and a~$3-i$-revealing edge means an absorbing state
		inside~$D$. 
		Note that $D'$ does not contain $3-i$-revealing state-actions
		neither since these would lead outside~$D$, and~$D'$ would not be
		$\delta_{3-i}$-closed.
		This means that the sub-MDP~$D'$ has the same support in both $\hat{M}_j$,
		hence it is also an end-component of~$\hat{M}_i$, hence~$D'$ is a double
		end-component.
		But this is only possible, by construction
		of~$\hat{M}$, if~$D=D'$ is an absorbing state.
	\end{proof}

	The following lemma refines the above one.
	\begin{lemma}
		\label{lemma:mec-transient-K}
		For any~$M$, and $\epsilon>0$, define $K=n\lceil
		\frac{\log(\epsilon)}{\log(1-p^n)}\rceil$,
		where~$p$ is the smallest nonzero probability of~$M$, and~$n$ the number of
		states.
		for any end-component $D$ of~$\hat{M}_i$ that is not a DEC,
		and any history~$h \in \calH(\hat{M})$ which contains a factor of length $K$
	compatible with~$D$, $\pr_{\hat{M}_{3-i},s}^\tau[h] \leq \epsilon$ for any
	strategy~$\tau$ and state~$s$.
\end{lemma}
\begin{proof}
	We know that $D$ does not contain an end-component in~$\hat{M}_{3-i}$.
	If $p$ denotes the smallest nonzero probability in~$\hat{M}$,
	then from any state~$s \in D$, the probability of leaving~$D$ after $n$ steps
	is at least $p^n$ under any strategy. So in~$K$ steps, the probability of
	leaving~$D$ is at least
	\(
		\sum_{i=0}^{K/n} (1-p^n)^ip^n = \frac{1- (1-p^n)^{K/n+1}}{p^n}p^n
		= 1 - (1-p^n)^{k+1}.
		\)
	which is at least $1-\epsilon$.
\end{proof}

In order to prove the ``equivalence'' of~$M$ and~$\hat{M}$ for objectives of
interest, we define a correspondance between histories of $M$ and~$\hat{M}$
which is, roughly, the projection defined by our transformation.
We distinguish the set $\calT(\hat{M}) = \{s_D \mid D$ $\text{distinguishing}\}$.
For any history $h= s_1a_1s_2a_2\ldots s_n \in \mathcal{H}(M)$,
let us define $\red(s_1a_1s_2a_2\ldots s_n) \in \mathcal{H}(\hat{M})$ by applying the
following transformations until a fixpoint is reached:

\begin{enumerate}
	\item If $h$ contains a state of~$\hatabstr^{-1}(\calT(\hat{M}))$, then if $i$ denotes the least index
	with $s_i \in \hatabstr^{-1}(\calT(\hat{M}))$, we remove the suffix $a_is_{i+1}\ldots s_n$.

	\item For any non-distinguishing MDEC~$D$, let $s_ia_i\ldots
	s_{i+k}$ be a maximal factor made of the states of~$D$. We remove from this
	factor all non frontier actions and states that precede. We project all states
	to~$s_D$, and any action~$a_{\alpha_j}$ from state~$s_{\alpha_j}$ to action
	$(s_{\alpha_j}a_{\alpha_j})$. We obtain a run of the form
	$s_D(s_{\alpha_1}a_{\alpha_1})s_D\ldots s_D (s_{\alpha_m}a_{\alpha_m})$ where each
	$s_{\alpha_i}$ is a frontier state, and $a_{\alpha_i}$ a frontier action
	from~$s_{\alpha_i}$.
\end{enumerate}

Let~$\calH_\calT(\hat{M})$ denote the histories of~$\hat{M}$ which does not
contain $\calT(\hat{M})$ except possibly on the last state.
The following lemma establishes the relation between~$M$ and~$\hat{M}$.

\begin{lemma}
	\label{lemma:wtsigma}
	For any MEMDP~$M$, state~$s$, strategy~$\sigma$, 
	there exists a strategy~$\hat{\sigma}$ such that
	for any history $h \in \calH_{\calT}(\hat{M})$,
	and any non-distinguishable MDEC~$D$,
	we have 
			$\pr_{\hat{M}_j,\hat{s}}^{\hat{\sigma}}(h)=
			\pr_{M_j,s}^\sigma[\red^{-1}(h) ],
			\pr_{\hat{M}_j,\hat{s}}^{\hat{\sigma}}(ha)=
			\pr_{M_j,s}^\sigma[\red^{-1}(ha) ]$, and
			$\pr_{\hat{M}_j,\hat{s}}^{\hat{\sigma}}[h a_D^\$]=
			\pr_{M_j,s}^\sigma[\red^{-1}(h D^\omega)].
			$
\end{lemma}
\begin{proof}
	Let us restate the equalities we are going to prove.
	\begin{equation}
		\label{eqn:wtinduct}
	\begin{array}{l}
			\pr_{\hat{M}_j,\hat{s}}^{\hat{\sigma}}(h)=
			\pr_{M_j,s}^\sigma[\red^{-1}(h) ],\\
			\pr_{\hat{M}_j,\hat{s}}^{\hat{\sigma}}(ha)=
			\pr_{M_j,s}^\sigma[\red^{-1}(ha) ],\\
			\pr_{\hat{M}_j,\hat{s}}^{\hat{\sigma}}[h a_D^\$]=
			\pr_{M_j,s}^\sigma[\red^{-1}(h D^\omega)].
	\end{array}
	\end{equation}

	Given~$\sigma$, we define $\hat{\sigma}$ as follows.
	For any history ending in~$\calH_\calT(\hat{M})$, $\hat{\sigma}$ is defined
	trivially.
	For any $a \in \hat{A}(h_i) \setminus \{a_D^\$\}_D$, define
	\[
		\hat{\sigma}(a \mid h_1\ldots h_i) = 
		\pr_{M_j,s}^\sigma(\red^{-1}(h_1\ldots h_ia) \mid
				\red^{-1}(h_1\ldots h_i)),
	\]
	for an arbitrary~$j$. These quantities do not depend on~$j$. In fact,
	$\pr_{M_j,s}^\sigma[\red^{-1}(h a) \mid \red^{-1}(h)]
	= \sum_{\pi \in \red^{-1}(h)} \pr_{M_j,s}^\sigma[\pi a \mid \pi]
		\pr_{M_j,s}^\sigma[\pi \mid h]
	= \sum_{\pi \in \red^{-1}(h)} \sigma(a\mid \pi)
		\pr_{M_j,s}^\sigma[\pi \mid h]$,
	and the latter factor does not depend on~$j$; since $\red^{-1}(h)$
	determines all outcomes of the actions whose distributions differ in
	both~$M_i$, and the distributions are identical in the remaining
	non-distinguishing double end components.

	For any $h_i = s_D$, where~$D$ is a non-distinguishing component, we
	let
	\[
		\hat{\sigma}(a_D^\$ \mid h_1\ldots h_i) = 
		\pr_{M_j,s}^\sigma(\red^{-1}(h_1\ldots h_iD^\omega) \mid \red^{-1}(h_1\ldots
					h_i)).
	\]
	We check that $\hat{\sigma}$ defines a probability distribution on
	available actions at any given history. For any state~$h_i\neq s_D$,
	probabilities $\hat{\sigma}(a \mid h_1\ldots h_i)$ clearly sum to~$1$ for~$a
	\in A(h_i)$. If $h_i = s_D$ for some non-distinguishing losing~$D$, any run
	that extends $h_1\ldots h_i$ either stays forever in~$D$, or takes one of the
	frontier actions for the first time. By definition, the former is the probability
	of~$\hat{\sigma}$ of choosing $a_D^\$$, and the latter that of choosing
	each~frontier action.
	
	We will prove \eqref{eqn:wtinduct} by induction on $i\geq 1$.

	For $i=1$, we have 
	$\pr_{\hat{M}_j,\hat{s}}^{\hat{\sigma}}(h_1)=
	\pr_{M_j,s}^\sigma[\red^{-1}(h_1)]$, which is~$1$ if $\hat{s}=h_1$
	and~$0$ otherwise.
	Furthermore, $\pr_{\hat{M}_j,\hat{s}}^{\hat{\sigma}}(h_1a_1)=
	\pr_{\hat{M}_j,\hat{s}}^{\hat{\sigma}}(h_1)\hat{\sigma}(a_1\mid
			h_1) = \pr_{M_j,s}^\sigma[\red^{-1}(h_1)]\cdot \pr_{M_j,s}^\sigma[\red^{-1}(h_1a_1)
				\mid \red^{-1}(h_1)] = \pr_{M_j,s}^\sigma[\red^{-1}(h_1a_1)]$.

	For $i>1$, we have
	\[
	\begin{array}{ll}
		\pr_{\hat{M}_j,\hat{s}}^{\hat{\sigma}}(h_1\ldots h_i)
		&= \pr_{\hat{M}_j,\hat{s}}^{\hat{\sigma}}
		[h_1\ldots h_{i-1}a_{i-1}] \hat{\delta}_j(h_{i-1},a_{i-1},h_i)
		\\&= \pr_{{M}_j,{s}}^{{\sigma}}
		[\red^{-1}(h_1\ldots h_{i-1}a_{i-1}) ]
				\hat{\delta}_j(h_{i-1},a_{i-1},h_i)
				\\&= \pr_{M_j,s}^\sigma[\red^{-1}(h_1\ldots h_{i-1}a_{i-1}h_i)].
	\end{array}
	\]
	The second line follows by induction, and the third line by definition 
	(explain).
	We have
	\[
	\begin{array}{ll}
	\pr_{\hat{M}_j,\hat{s}}^{\hat{\sigma}}(h_1\ldots h_i a_i)
	&= \pr_{\hat{M}_j,\hat{s}}^{\hat{\sigma}}(h_1\ldots h_i)
		\hat{\sigma}(a_i\mid h_1\ldots h_i)
		\\&=\pr_{M_{j},s}^\sigma[\red^{-1}(h_1\ldots h_i) ]
	\\&~\cdot
		\pr_{M_{j},s}^\sigma[\red^{-1}(h_1\ldots h_i a_i)  \mid
		\red^{-1}(h_1\ldots h_i)] 
	\end{array}
	\]
	The third equality is proved similarly.
\end{proof}


The equivalence between~$M$ and~$\hat{M}$ for reachability and safety objectives
is obtained as the following corollary. Note that the value vectors
are preserved although vectors achieved in~$\hat{M}$ may not be achievable
in~$M$.
\begin{corollary}
	\label{corollary:Mtilde}
	For any MEMDP~$M$, and $\Phi$
	a reachability or safety objective, $\Val^*_{\Phi}(M,s) =
	\Val^*_{\hat{\Phi}}(\hat{M},\hat{s})$. 
\end{corollary}

By Definition~\ref{def:tilde}, and the previous corollary, we assume, in the
next section, that the MEMDPs we consider have only trivial DECs.
\begin{assumption}
	\label{assum:only-trivial-dec}
	All MEMDPs are assumed to have only trivial DECs.
\end{assumption}

\section{Limit-Sure Reachability}
\label{section:limitsure}
In this section, we give a polynomial-time algorithm for limit-sure reachability in
MEMDPs.
For any MEMDP~$M$, and reachability objective $\Phi$,
we define the set of limit-sure winning states~$W(M,\Phi)$ as follows.
We have $s \in W(M,\Phi)$ if either $s \in R^{\Phi}$, or there exists a family of
strategies witnessing limit-sure satisfaction, that is, for any~$\epsilon>0$, 
there a strategy~$\sigma_\epsilon$ such that
$\pr_{M_i,s}^{\sigma_\epsilon}[\Phi]\geq 1-\epsilon$ for
$i=1,2$.

The following lemma states an important property of the set~$W(M,\Phi)$ for
reachability objectives but also safety objectives.
\begin{lemma}
	\label{lemma:only11-inf}
	On any MEMDP~$M$, and a reachability or safety objective~$\Phi$,
	there exists a memoryless strategy~$\sigma_W$ under which  
	from any~$s \in W(M,\Phi)$, each $M_i$ stays surely inside~$W(M,\Phi)$.
\end{lemma}
\begin{proof}[of Lemma~\ref{lemma:only11-inf}]
	In this proof only, we separate control and probabilistic states for convenience.
	Given a state~$s$, and action $a \in A(s)$, we denote by $sa$ the intermediate
	probabilistic state reached by chosing action~$a$.
	We denote $W = W(M,\Phi)$.

	We show that all successors of probabilistic states~$sa \in W$ are in~$W$.
	In fact, assume that there exists~$s'\not \in W$ such that
	$\delta_i(s,a,s')\neq 0$ for some~$i$.
	This means that $s' \not \in R^{\Phi}$ and there is no family of strategies
	witnessing limit-sure winning from~$s'$.
	If $s' \in R\setminus R^{\Phi}$, then 
	there exists $\epsilon_0>0$ such that
	for any strategy~$\sigma$, 
	$\pr_{M_i,s'}^\sigma[\Phi] \leq 1- \epsilon_0$,
	therefore
	$\pr_{M_i,s}^\sigma[\Phi] \leq 1 - \delta_i(s,a,s') +
	\delta_i(s,a,s')(1-\epsilon_0) \leq 1- \epsilon_0 \delta_i(s,a,s')$ contradicting that $s \in W$.
	Note that we cannot have $s' \in R_{3-i}$ since $\delta_i(s,a,s')\neq 0$.
	Now, if $s'$ is unrevealed then $\delta_j(s,a,s') \neq 0$ for both~$j=1,2$.
	By assumption that $s' \not \in W$, there exists $\epsilon_0>0$ such that for any
	strategy~$\sigma$, $\Val^*_{\Phi}(M_j,s')\leq 1- \epsilon_0$ for some~$j$.
	Then, for any~$\sigma$, for some~$j$,
	$\pr_{M_j,s}^\sigma[\Phi]\leq 1-\delta_j(s,a,s')\epsilon_0$
	contradicting $s \in W$.

	We now prove that for any control state~$s \in W$, there exists an action~$a$
	such that $\supp(s,a) \subseteq W$, by induction on the length $k>0$ of the
	history. At the same time, we define the strategy~$\sigma_W$ by setting
	$\sigma_W(s) = a$.

	The case $k=1$ is trivial since $s \in W$.
	For probabilistic states~$sa$, the property follows from the above paragraph.
	Assume $k\geq 2$. If there exists~$a \in A(s)$ such
	that $sa \in W$, then by induction hypothesis, for
	all~$\epsilon>0$, there exists a strategy~$\sigma'$
	that witnesses $1-\epsilon$-satisfaction from the (probabilistic) state $sa$
	and stays in $W$ states for~$k-1$ steps. We let~$\sigma_W(s) = a$.

	We now prove that there must exist such an action~$a$. To get a contradiction,
	assume that for all actions~$a \in A(s)$, $sa \not \in W$. This
	means that for all~$a \in A(s)$, there exists $\epsilon_a > 0$ such that
	$\Val^*_{\Phi}(M_j,sa) \leq 1- \epsilon_a$ for some~$j$.
	Let~$\sigma$ be a strategy witnessing $1-\epsilon$-satisfaction
	for $\epsilon < \frac{1}{|A(s)|}\min_{a \in A(s)}\epsilon_a$. 
	There exists~$a \in A(s)$ such that from~$s$, $\sigma$ assigns a probability
	of at least $\frac{1}{|A(s)|}$ to~$a$. Let~$j$ such that
	$\Val^*_{\Phi}(M_j,sa) \leq 1- \epsilon_a$. We have
	$\pr_{M_j,s}^\sigma[\Phi]
		\leq \frac{1}{|A(s)|}(1-\epsilon_a) + 1 - \frac{1}{|A(s)|}
		< 1- \epsilon$, contradiction.
\end{proof}

In the rest of the paper, $\sigma_W$ will denote the pure memoryless strategy of
Lemma~\ref{lemma:only11-inf}. Note that we do not require the
\emph{computability} of $\sigma_W$ at this point.

In the rest of this section, we assume, by
Assumption~\ref{assum:only-trivial-dec}, that the considered MEMDPs have
only trivial DECs.

\begin{figure}[h]
	\begin{tikzpicture}
		\tikzstyle{every state}=[node distance=1cm,minimum size=10pt, inner sep=1pt];
		\node[state] at (0,0) (A){$s$};
		\node[state,right of=A,fill=black,minimum size=3pt] (Aa) {};
		\node[state,right of=Aa] (B){$t$};
		\node[state, node distance=0.7cm, above of=B] (C){$u$};
		\node[state, below of=A] (D){$v$};
		\node[state, right of=D] (E){$w$};
		\node[right of=C, node distance=0.5cm] {$T$};
		\node[right of=E, node distance=0.5cm] {$T$};
		\path[-latex'] (A) edge node[above]{$a$} (Aa)
			(Aa) edge (B)
			(Aa) edge[dashed, bend left] (C)
			(B) edge[bend left] node[below,pos=0.1]{$a$} (A)
			(C) edge node[right]{$a$} (B)
			(A) edge node[left]{$b$} (D)
			(D) edge node[below]{$a$}(E)
			(E) edge[loop above] (E);
		\node[below of=D, anchor=west, node distance=0.7cm]{\scriptsize MEMDP~$M$};

		\begin{scope}[shift={(4,0)}]
			\node[state] at (0,0) (A){$s$};
			\node[state,right of=A,fill=black,minimum size=3pt] (Aa) {};
			\node[state,right of=Aa] (B){$t$};
			\node[state, node distance=0.7cm, above of=B] (C){$u$};
			\node[state, below of=A] (D){$v$};
			\node[state,right of=D,fill=black,minimum size=3pt] (Dd) {};
			\node[state, right of=Dd] (E){$w$};
			\node[state, below of=E, node distance=0.4cm] (F){};
			\node[right of=C, node distance=0.5cm] {$T$};
		\node[right of=E, node distance=0.5cm] {$T$};
			\path[-latex'] (A) edge node[above]{$a$} (Aa)
				(Aa) edge (B)
				(Aa) edge[dashed, bend left] (C)
				(B) edge[bend left] node[below,pos=0.1]{$a$} (A)
				(A) edge node[left]{$b$} (D)
				(D) edge node[below]{$a$}(Dd)
				(Dd) edge[dashed] (F)
				(Dd) edge[dotted] (E)
				(C) edge[loop above] (C)
				(E) edge[loop above] (E)
				(F) edge[loop right] (F);
		\node[below of=D, anchor=west, text width=2.5cm, node distance=1cm]{\scriptsize
			Revealed form and absorbing targets};
		\draw[-,gray,node distance=1cm] ($(A.north west)+(-0.2,0.2)$) rectangle
		($(B.south east)+(0.2,-0.3)$) node[above right,black] {MGEC~$D$};
		\end{scope}
		\begin{scope}[shift={(8,0)}]
			\node at (0,0) (A){};
			\node[right of=A,minimum size=3pt] (Aa) {};
			\node[right of=Aa] (B){};
			\node[state, node distance=0.7cm, above of=B] (C){$u$};
			\node[state, below of=A] (D){$v$};
			\node[state,right of=D,fill=black,minimum size=3pt] (Dd) {};
			\node[state, right of=Dd] (E){$w$};
			\node[state, below of=E, node distance=0.4cm] (F){};
			\node[right of=C, node distance=0.5cm] {$\wt{T}$};
			\node[right of=E, node distance=0.5cm] {$\wt{T}$};
			\node[state] at(1,0) (TD) {$t_D$};
			\node[right of=TD] {$\wt{T},\wt{T}$};
			\path[-latex'] 
				(D) edge node[below]{$a$}(Dd)
				(Dd) edge[dashed] (F)
				(Dd) edge[dotted] (E)
				(C) edge[loop above] (C)
				(E) edge[loop above] (E)
				(F) edge[loop right] (F)
				(TD) edge[loop above] (TD);
			\node[below of=D, anchor=west, node distance=0.7cm]{\scriptsize MEMDP~$\wt{M}$};
		\end{scope}
	\end{tikzpicture}
	\caption{\small On the left, an MEMDP with 
		objective~$\objreach(T)$, which is \emph{not} in revealed form; an
		equivalent instance~$M$ in revealed form is shown in the middle.
		Note that~$M$ has only trivial DECs.
		States~$\{s,t\}$ induce a good end-component~$D$ in~$M_2$; in fact, the strategy
		choosing action~$a$ at~$s$ and~$t$ is almost surely winning in~$M_1$.
		The construction~$\wt{M}$ is shown on the right,
		where all states of~$D$ are contracted as~$t_D$ which
		becomes a target state. Because $\wtabstr(s) = t_D$,
		objective~$\Phi$ is achieved limit-surely from~$s$.
	}
	\label{fig:lsreach}
\end{figure}
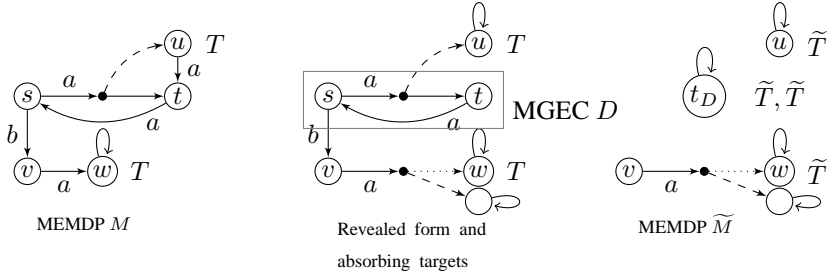

Let us explain the idea behind the limit-sure reachability algorithm on the
MEMDP~$M$ of Fig.~\ref{fig:lsreach}. Here, the MDP~$M_1$ has a MEC~$D$ with the
following property: the strategy~$\sigma$ compatible with~$D$ and choosing
all actions of~$D$ uniformly at random, achieves the objective almost surely
in~$M_2$. In fact, a strategy that chooses~$a$ at states~$s$ and~$t$ almost
surely reaches~$u$ in~$M_2$. On order to achieve the objective with
probability close to~$1$, one can run strategy~$\sigma$ for a large number of
steps, and if the objective is still not achieved, switch to the optimal
strategy for~$M_1$, which consists in choosing~$b$ from~$s$.
It can be shown that such a strategy achieves the objective at probabilites
$(1-\epsilon,1-\epsilon)$, for any desired~$\epsilon>0$, from any state of such
end-components.
Our algorithm consists in identifying these end-components and contracting them
as winning absorbing states.


Formally, let an end-component~$D$ of~${M}_i$ be called \emph{good} if the strategy that chooses
all edges of~$D$ uniformly at random is almost sure winning for~${M}_{3-i}$, from any
state. 
Observe that the union of good end-components with a non-empty intersection
is a good end-component. We
will thus consider \emph{maximal good end components (MGECs)} which can be
computed in polynomial time as follows.

\begin{lemma}
	\label{lemma:mgec}
	Let~$M$ be a MEMDP with only trivial DECs, and a reachability objective
	$\Phi$.
	For any~$i=1,2$, consider the set 

		$U_i=\{s \mid \exists D \in \mecs(M_i), s \in D, \Val_{\Wsafety(D\cup
		R_{3-i}^{\Phi})}^*(\restr{\cup M}{D\cup R_{3-i}},s)=1\}.$

	Let~$M_i'$ denote the sub-MDP of~$M_i$ induced by~$U_i$.
	Then the MGECs of~$M_i$ are the union of the MECS of $M_i'$,
	and the trivial MECs of~$M_i'$ surely satisfying~$\Phi$.
\end{lemma}
\begin{proof}
	To see that the sub-MDP $M_i'$ is well-defined, notice that for each~$D$, the
	states satisfying the safety condition induces a sub-MDP, and that these sub-MDPs
	are disjoint for each~$D$.

	Let us show that non-trivial MECs of~$M_i'$ and trivial-and-winning MECs
	of~$M_i'$ are MGECs of~$M_i$. Note that we distinguish here the case of
	trivial MECs since our definition $U_i$ could yield trivial MECs that are not
	winning.
	It is clear that trivial MECs of~$M_i'$ satisfying~$\Phi$ are maximal good
	end-components.
	Consider a non-trivial MEC~$G$ of~$M_i'$. Let~$\tau$ be the uniform strategy
	inside~$G$ in~$M_i'$. Clearly, $\tau$ stays inside~$G$ in~$M_i$. In~$M_{3-i}$,
	we know that strategy~$\tau$ leaves~$G$ almost surely by
	Lemma~\ref{lemma:mec-transient}.
	But by Assumption~\ref{assum:revealed}, and by the fact that~$\tau$ is
	compatible with~$U$, $\tau$ also ensures $\Wsafety(D\cup
	R_{3-i}^{\Phi})$ surely, so~$R_{3-i}^{\Phi}$ must be reached
	almost surely in~$M_{3-i}$. Therefore, $G$ is a good end-component. We will
	show its maximality at the end of this proof.

	Conversely, we show that MGECs of~$M_i$ are MECs of~$M_i'$s.
	Any MGEC~$G$ of~$M_i$ is in particular a MEC of~$M_i$, so it is included in some
	$D \in \mecs(M_i)$.  Let~$\tau$ be the uniform strategy in~$G$.
	Clearly, we have $\pr_{M_i,s}^\tau[ \Wsafety(D)]=1$ for any~$s \in D$,
	and $\pr_{M_{3-i},s}^\tau[\Wsafety(D\cup R_{3-i}^{\Phi})]=1$. In fact,
	because strategy~$\tau$ is compatible with~$D$ in~$M_i$, and by
	Assumption~\ref{assum:revealed}, any action of~$D$ which leaves~$D$
	in~$M_{3-i}$ ends in~$R_{3-i}$. Furthermore, because~$\tau$ is almost surely
	winning for~$M_{3-i}$ from~$D$, we have that $\Wsafety(D \cup
	R_{3-i}^{\Phi})$ holds surely in~$M_{3-i}$ under~$\tau$.
	It follows that~$G$ is included in~$M_i'$. Moreover, 
	$G$ is by definition an end-component in~$M_i'$. 
	To show that~$G$ is maximal, 
	assume that there exists $G \subsetneq G' \subseteq M_i'$ where~$G'$ is a MEC
	in~$M_i$. By the first case, $G'$ is a good end-component which contradicts
	the maximality of~$G$ as a good end-component. Therefore, $G$ is indeed a MEC
	of~$\restr{M_i}{U}$.

	To finish the proof, we show that a non-trivial MEC~$G$ of~$M_i'$ is a
	\emph{maximal} good end-component.
	Towards a contradiction, assume that there exists~$G \subsetneq G'$ a MGEC
	of~$M_i$. By the second case above, $G'$ is then a MEC of~$M_i'$ which
	contradicts the maximality of~$G$ as an end-component of~$M_i'$.
\end{proof}

\begin{definition}[Transformation~$\wt{M}$]
	\label{def:wtM}
	Given any MEMDP~$M$ with only trivial DECs, and reachability
	objective $\Phi$.
	we define $\wt{M}=(\wt{S},\wt{A},\wt{\delta_1},\wt{\delta_2})$ by applying the following transformation to~${M}$.
	Mark any state~$s$ that belongs to some MGEC~$D$ of~$M_i$ for some~$i=1,2$,
	by $D$. If a state can be marked twice, choose one marking arbitrarily.
	We define~$\wt{M}$ by redirecting any edge entering a state marked by some
	$D$ to a fresh absorbing state~$t_D$.
	For each~$i=1,2$, the reachability objective $\wt{\Phi}$ is defined by the
	union of~$\Phi$, with all states~$t_D$ such that $\Phi$ can be ensured
	almost surely from~$D$ in~$M_i$.
\end{definition}
Let us denote by~$\wtabstr(\cdot)$ the mapping from the states~$M$ to those
of~$\wt{M}$.

The following lemmas establish the equivalence between limit-sure objectives
in~$M$ and corresponding almost-sure objectives in~$\wt{M}$. The algorithm for
limit-sure objectives is then obtained by using the algorithm of
Section~\ref{section:reachsafe}.
Note that only the first lemma is constructive, but it is the one that we need
to compute strategies for~$M$.

\begin{lemma}
	\label{lemma:lswtM}
	For any MEMDP~$M$ with only trivial DECs, and reachability objective
	$\Phi$,
	if $\wt{\Phi}$ can be achieved almost surely
	in~$\wt{M}$, then $\Phi$ can be achieved limit surely in~$M$.
	Moreover, given an almost sure winning strategy for~$\wt{M}$, for any~$\epsilon>0$,
	a strategy with memory $O(\frac{\log(\epsilon)}{\log(1-p)})$ for~$M$,
	where~$p$ is the smallest nonzero probability, achieving probabilities $1-\epsilon$ can
	be computed.
\end{lemma}
\begin{proof}
	Let~$\sigma$ be a strategy achieving each $\wt{\Phi}$ almost surely
	in~$\wt{M}$. For any~$\epsilon>0$, we derive a strategy for~${M}$
	achieving $\Phi$ with probability $1-\epsilon$ for each~$M_i$.
	For this, we define $\sigma_\epsilon$ for~${M}$ by modifing~$\sigma$ as
	follows. Remember that all target states are absorbing by
	Assumption~\ref{assum:absorbing}.
	Fix any~$\epsilon \in [0,1]$, and let~$p$ be the smallest nonzero probability
	in~$M$. Define $K\geq \frac{\log(\epsilon)}{\log(1-p)}$.
	Upon arrival to any state of a MGEC
	$D$ of~${M}_j$, if~$D$ is a trivial DEC, then we extend the strategy
	trivially. Otherwise, we switch to a strategy $\tau$ compatible with~$D$ in~${M}_j$
	picking all actions in~$D$ uniformly at random. Note that under this strategy,
	actions $B$ that leave~$D$ in~$M_{3-j}$ with positive probability are seen
	infinitely often (since $D$ is not a DEC). These actions lead to
	$3-j$-revealed states in~$M_{3-j}$ from which the strategy is extended
	trivially. Whenever actions in~$B$ are seen~$K$ times, if the play is still in~$D$ then
	we switch to the optimal strategy for~${M}_j$. 
	Notice that the probability of staying inside~$D$ under~$\tau$ in~$M_j$
	is~$1$, while the probability of leaving~$D$ in~$M_{3-j}$ under strategy is at
	least $1-\epsilon$ by the choice of~$K$.

	Assume $\wt{\Phi} = \objreach(T\cup \{t_D\})$.
	Because $\wt{\Phi}$ is ensured almost surely in~$\wt{M}_i$, in~$M_i$
	under~$\sigma$, almost surely we either reach~$T$ or switch to~$\tau$.
	The claim follows since when we switch to~$\tau$, $T$ is reached with
	probability at least~$1-\epsilon$.
\end{proof}
\begin{lemma}
	\label{lemma:lsM}
	Let~$M$ be any MEMDP  with only trivial DECs, and
	$\Phi$ reachability objectives. Let~$\sigma_W$ denote the strategy of
	Lemma~\ref{lemma:only11-inf} for~$M$, and~$\wt{\sigma}_W$ obtained
	from~$\sigma_W$ by extending it trivially on states~$t_D$.
	For any~$s \in W(M,\Phi)$, $\Val_{\wt{\Phi}}^{\wt{\sigma}_W}(\wt{M},\wt{s})=1$.
\end{lemma}
\begin{proof}[of Lemma~\ref{lemma:lsM}]
	For strategy~$\wt{\sigma}_W$ and starting state~$\wt{s}$, 
	let~$D$ be any MEC of~$\wt{M}_i$ in which the play stays forever with positive
	probability. We have $D \subseteq \wtabstr(W(M,\Phi))\cup \{t_D\}_D$ since $\sigma_W$
	does not leave the set $W(M,\Phi)$ in~$M$. If~$D$ is a DEC, then it is trivial
	and satisfies the objective. 
	If~$D$ is not a DEC, then it is transient in
	$\wt{M}_{3-i}$. But because~$\sigma_W$ does not leave the set $W(M,\Phi)$, all
	revealed states reached under~$\sigma_W$ from~$D$ in~$\wt{M}_{3-i}$ are in
	$R^{\Phi}_{3-i}$, therefore winning. It follows that~$D$ is a good
	end-component, contradiction since all such components were reduced
	in~$\wt{M}$. Therefore, any MEC $D$ of~$\wt{M}_i$ in which the play stays
	forever is a DEC satisfying~$\wt{\phi}_i$. The lemma follows.
\end{proof}

The algorithm consists in constructing~$\wt{M}$ and solving almost-sure
reachability for~$\wt{\Phi}$:
\begin{theorem}
	\label{thm:lsreach}
	The limit-sure reachability problem is decidable in polynomial-time.
\end{theorem}

\section{Quantitative Reachability}
\label{section:quantitative}
We are now interested in the general quantitative reachability problem for MEMDPs.
We first show that the problem is NP-hard, so it is unlikely to
have a polynomial-time algorithm, and techniques based on linear
programming cannot be applied. We will then derive an approximation algorithm.

\subsection{Hardness}
\label{section:hardness}
We prove the following theorem.

\begin{theorem}
	\label{thm:nphard}
	Given an MEMDP~$M$, target set~$T$, and $\alpha_1,\alpha_2 \in [0,1]$, it is NP-hard to decide
	whether for some strategy~$\sigma$,
	$\pr_{M_i,s_0}^\sigma[\objreach(T)] \geq \alpha_i$ for each~$i=1,2$.
\end{theorem}

The following \textsf{Product-Partition} problem is NP-hard in the strong sense.
Given positive integers $v_1,\ldots,v_n$,
decide whether there exists a subset $I
\subseteq \{1,\ldots,n\}$ such that $\prod_{i \in I} v_i = \prod_{i \not \in I}
v_i$.
It is easy to see that the problem is equivalent if the target value
$\sqrt{v_1\ldots v_n}$
is given as part of input. In fact, if~$l$ is the maximum number of bits
required to represent any~$v_i$, then~$V=v_1\ldots v_n$ can be computed in time
$n^2l^2$. Further, one can check if~$V$ is a perfect square and (if it is)
compute the square root in time $O(\log(V))=O(nl)$ by binary search.

We reduce this problem to quantitative reachability in MEMDPs.
	We fix an instance of the problem, and construct the following MEMDP~$M$.
	\begin{figure}[h!]
		\begin{center}
			\begin{tikzpicture}
			\tikzstyle{every state}=[minimum size=3pt,inner sep=0pt,node distance=2cm]
			\tikzstyle{every node}=[font=\small]
			\node[state] at (0,0) (s1) {$s_1$};
			\node[below right of=s1,node distance=40pt] (s1q){};
			\node[state,right of=s1,node distance=3cm] (nop) {$\bot$};
			\fill (s1q) circle(2pt);
			\node[state,below of=s1] (s2) {$s_2$};
			\node[below right of=s2,node distance=40pt] (s2q){};
			\fill (s2q) circle(2pt);
			\node[state, below of=s2] (s3) {$s_3$};
			\node[below of=s3] (s4) {$\vdots$};
			\node[below right of=s4,node distance=20pt] (s4q){};
			\fill (s4q) circle(2pt);
			\node[state,below of=s3,accepting] (sn) {$s_{n+1}$};
			\path[draw,-latex'] 
			(s1) edge[bend left] node[right]{$a$} (s1q) 
			(s1q) edge[bend left] node[right]{$\frac{1}{v_1}$} (s2) 
			(s1q) edge[bend right] node[left]{$1-\frac{1}{v_1}$} (nop) 
			(s1) edge[bend right] node[left]{$b$} (s2) 
			(s2) edge[bend left] node[right]{$a$} (s2q)
			(s2q) edge[bend left] node[below]{$\frac{1}{v_2}$} (s3) 
			(s2q) edge[bend right] node[left]{$1-\frac{1}{v_2}$} (nop) 
			(s2) edge[bend right] node[left]{$b$} (s3)
			(s4) edge[bend left] node[right]{$a$} (s4q)
			(s4) edge[bend right] node[left]{$b$} (sn)
			(s4q) edge[bend right] node[left]{$1-\frac{1}{v_n}$} (nop)
			(s4q) edge[bend left] node[right]{$\frac{1}{v_n}$} (sn)
            (nop) edge[loop right] node[right]{$a$} (nop)
            (sn) edge[loop below] (sn);
		\end{tikzpicture}
		\end{center}
	\end{figure}
	The figure depicts the MDP~$M_1$, while $M_2$ is obtained by inversing the roles
	of~$a$ and~$b$. We let~$s_{n+1}$ be the target state, and define~$T = \{s_{n+1}\}$.
	Let us denote $W = 1/V$.
	We will prove that $M$ has a strategy achieving the probabilities
	$(\sqrt{W},\sqrt{W})$ for reaching~$s_{n+1}$ 
	if, and only if the \textsf{Product-Partition} problem has a solution.
	Notice that the reduction is polynomial since all probabilities can be encoded
	in polynomial time.

	Observe that to each pure strategy~$\sigma$ corresponds
	a set $S_\sigma=\{ i \mid \sigma(s_i,b) = 1\}$.
	We have that $\pr_{M_2,s_1}^\sigma[\objreach(T)] = \prod_{i \in
	S_\sigma}\frac{1}{v_i}$, and
	$\pr_{M_1,s_1}^\sigma[\objreach(T)] = \prod_{i \not\in
	S_\sigma}\frac{1}{v_i}$
	Therefore, a pure strategy with values 
	$(\sqrt{W},\sqrt{W})$ yields a solution to the \textsf{Product-Partition} problem, and
	conversely. 
	To establish the reduction, we need to show that if some arbitrary strategy achieves the probability vector
	$(\sqrt{W},\sqrt{W})$ in~$M$, then there is a pure strategy achieving the same
	vector. 
	\medskip

	To ease reading, for any strategy~$\sigma$, let us denote $p^\sigma_i =
	\pr_{M_i,s_1}^\sigma[\objreach(T)]$.
	Let $\Sigma^D$ denote the set of deterministic strategies.

    \begin{lemma}
      \label{lemma:lincombin}
      For any strategy~$\sigma$, there exists
      $(\lambda_\pi)_{\pi \in \Sigma^D}$ with $0\leq \lambda_\pi \leq 1$
      and $\sum_{\pi\in \Sigma^D} \lambda_\pi = 1$ such that
      $p_i^\sigma = \sum_{\pi\in \Sigma^D}\lambda_\pi p^\pi_i$ for
      all~$i=1,2$.
    \end{lemma}
    \begin{proof}
      Consider any strategy~$\sigma \colon (SA)^*S \rightarrow \calD(A)$. Observe that since~$\bot$ is an absorbing state,
      $\sigma$ is characterized by the choices at histories not ending in~$\bot$, that is, histories that belong to
      $s_1(a+b)s_2(a+b)\ldots (a+b)s_i$. 

      Similarly, a deterministic strategy is characterized by the unique sequence of actions it takes from~$s_1$ to~$s_{n+1}$ when it avoids~$\bot$.
      Accordingly, we will identify the words of~$(a+b)^n$ with deterministic strategies, and denote $p_i^\pi$ the probability of reaching~$T$
      in~$M_i$ under strategy~$\pi \in (a+b)^n$.

      Under strategy~$\sigma$, there are only $2^n$ histories that allow reaching the target state~$s_{n+1}$. We express
      this probability summing over the probabilities of all these histories.
      We have
      \[
      \begin{array}{ll}
        p_i^\sigma &= \sum_{\pi \in (a+b)^n} \prod_{i=1}^n \sigma(\pi_{i} \mid s_1\pi_1 \ldots s_i) \delta_i(s_{i-1},\pi_i,s_i)\\
        &=\sum_{\pi \in (a+b)^n} \big(\prod_{i=1}^n \sigma(\pi_{i} \mid s_1\pi_1 \ldots s_i)\big) \prod_{i=1}^n \delta_i(s_{i-1},\pi_i,s_i)\\
        &=\sum_{\pi \in (a+b)^n} \big(\prod_{i=1}^n \sigma(\pi_{i} \mid s_1\pi_1 \ldots s_i)\big) p_i^\pi\\
      \end{array}
      \]
      Let us set~$\lambda_\pi = \big(\prod_{i=1}^n \sigma(\pi_{i} \mid s_1\pi_1 \ldots s_i)\big)$.
      Hence, we have written~$p_i^\sigma$ as a linear combination of the reachability probabilities of deterministic strategies.

      It remains to show that the weights form a probability distribution, that is, $\sum_{\pi \in (a+b)^n}\lambda_\pi = 1$.
      Let $H = s_1 + s_1(a+b)s_2 + \ldots s_1(a+b)\ldots (a+b)s_n$.
      We will prove by induction that for any history $h \in H                                                 $,
      \[
      \sum_{\pi \in (a+b)^{n-\lfloor|h|/2\rfloor}}\prod_{i=1}^{|\pi|} \sigma(\pi_i \mid h \pi_1\ldots s_{\lfloor |h|/2\rfloor + i})
=1.
      \]
      This proves our claim by choosing $h = s_1$.
      We proceed backwards from~$|h|=2n-1$ down to~$1$. For~$|h|=2n-1$, the quotient set~$h^{-1}H$ is empty
      so the product is~$1$, and the equality holds.
      Consider any~$h$ with $|h| < 2n-1$.
      We write
      \[
      \begin{array}{l}
        \sum_{\pi \in (a+b)^{n-\lfloor|h|/2\rfloor}} \prod_{i=1}^{|\pi|} \sigma(\pi_i \mid h \pi_1\ldots s_{\lfloor |h|/2\rfloor + i})\\
        = 
        \sum_{x \in \{a,b\}} \sum_{\pi \in x(a+b)^{n-\lfloor|h|/2\rfloor}-1} \prod_{i=1}^{|\pi|} \sigma(\pi_i \mid h \pi_1\ldots s_{\lfloor |h|/2\rfloor + i})\\
        =\sum_{x \in \{a,b\}}  \sigma(x \mid h)\sum_{\pi \in x(a+b)^{n-\lfloor|h|/2\rfloor}-1} 
        \prod_{i=2}^{|\pi|} \sigma(\pi_i \mid h \pi_1\ldots s_{\lfloor |h|/2\rfloor + i})\\
        =\sum_{x \in \{a,b\}}  \sigma(x \mid h)\sum_{\pi \in (a+b)^{n-\lfloor|h'|/2\rfloor}} 
        \prod_{i=1}^{|\pi|} \sigma(\pi_i \mid h' \pi_1\ldots s_{\lfloor |h'|/2\rfloor + i})\\
        = 1.
      \end{array}
      \]
      where~$h' = hxs_{\lfloor h/2\rfloor+1}$.
      Here $|h'|>|h|$, so by induction, the inner sum is equal to~$1$ in the second to the last line. Moreover,
      $\sigma(a \mid h) + \sigma(b \mid h) = 1 $ for any history~$h$, which yields the last line,
      hence the claim.
    \end{proof}

	The following lemma is the last step of the reduction: if there is a strategy
	whose reachability probabilities are no greater than~$(\sqrt{W}+\epsilon,\sqrt{W}+\epsilon)$ component-wise,
	for some well chosen~$\epsilon$, then there is a pure strategy under which the reachability probabilities
    are exactly $(\sqrt{W},\sqrt{W})$. 
\begin{lemma}
	\label{lemma:hardness-epsilon}
	Given $v_1,\ldots,v_n \in \mathbb{Z}^+$, and $W = \prod_{i=1}^n
	\frac{1}{v_i}$, let $\epsilon < \frac{1}{4}\sqrt{W}$.
	If there exists a strategy~$\sigma$ such that for~$i=1,2$, $p_i^\sigma =  \sqrt{W}+\delta_i$
	for some $\delta_1,\delta_2 \in[-\sqrt{W},\epsilon]$, 
	then there is a pure strategy~$\pi$ such that $p_i^\pi = \sqrt{W}$ for all~$i=1,2$.
\end{lemma}
\begin{proof}
	Consider any~$\sigma$ with value vector~$(\sqrt{W}+\delta_1,\sqrt{W}+\delta_2)$.
	By Lemma~\ref{lemma:lincombin}, we write $\sigma$ as the linear combination of pure strategies
	as $\sigma = \sum_{i=1}^n \lambda_i \pi_i$, we get
	\[
		\begin{array}{l}
			\lambda_1 w_1 + \ldots + \lambda_n w_n = \sqrt{W}+\delta_1,\\
			\lambda_1 W/w_1 + \ldots + \lambda_n W/w_n = \sqrt{W}+\delta_2,\\
		\end{array}
	\]
	where $w_i = p^{\pi_i}_1$, and $W/w_i = p^{\pi_i}_2$.
	By dividing the second equation by~$W$, distributing the right hand side in
	the linear combination in both lines, and multiplying the second line by~$-1$, 
    we rewrite this as
	\[
		\begin{array}{l}
			\lambda_1 (w_1-\sqrt{W}) + \ldots + \lambda_n (w_n-\sqrt{W}) = \delta_1,\\
			\lambda_1 (\frac{w_1 - \sqrt{W}}{w_1\sqrt{W}}) + \ldots + \lambda_n
			(\frac{w_n-\sqrt{W}}{w_n\sqrt{W}}) + \frac{\delta_2}{W} = 0.\\
		\end{array}
	\]
	Towards a contradiction, assume that $w_i \neq \sqrt{W}$ for all~$i$.
	Define $P \subsetneq \{1,\ldots,n\}$, the set of~$i$
	such that $w_i - \sqrt{W}>0$, and let~$N = \{1,\ldots,n\} \setminus P$.
	For all~$i\in P$, we have $1/w_i \leq \sqrt{1/W}-1$ which means
	$w_i \geq \frac{\sqrt{W}}{1-\sqrt{W}}$. For~$i\in N$, we similarly obtain
	$w_i \leq \frac{\sqrt{W}}{1+\sqrt{W}}$.
	We obtain that for any $i \in P$, $\frac{1}{w_i\sqrt{W}} \leq \frac{1-\sqrt{W}}{W}$,
	and for any $i \in N$, $\frac{1}{w_i\sqrt{W}} \geq \frac{1+\sqrt{W}}{W}$.
	We rewrite
	\begin{equation}
		\label{eqn:NP}
		\begin{array}{l}
			\sum_{i \in P} \lambda_i (w_i - \sqrt{W}) - \sum_{i \in N} \lambda_i
			(\sqrt{W}-w_i) = \delta_1,\\
			\sum_{i \in P} \lambda_i \frac{w_i - \sqrt{W}}{w_i\sqrt{W}}
      =
			\sum_{i \in N} \lambda_i \frac{\sqrt{W}-w_i}{w_i\sqrt{W}}
-\frac{\delta_2}{W}
		\end{array}
	\end{equation}
	We have
	\[
		\begin{array}{l}
			\frac{1+\sqrt{W}}{W}\sum_{i \in N}\lambda_i (\sqrt{W}-w_i) -
      \frac{\delta_2}{W}\\
			\leq \sum_{i \in N} \lambda_i \frac{\sqrt{W}-w_i}{w_i\sqrt{W}} -
      \frac{\delta_2}{W}
			=
			\sum_{i \in P} \lambda_i \frac{w_i - \sqrt{W}}{w_i\sqrt{W}}\\
			\leq \frac{1-\sqrt{W}}{W} \sum_{i \in P}\lambda_i (w_i - \sqrt{W})
		\end{array}
	\]
	It folows that $(1+\sqrt{W})\alpha_N - (1-\sqrt{W})\alpha_P \leq \delta_2$,
	where~$\alpha_P = \sum_{i \in P}\lambda_i(w_i - \sqrt{W})$ and
	$\alpha_N=\sum_{i \in N}\lambda_i (\sqrt{W}-w_i)$; so we get $(\alpha_N-\alpha_P) +
	\sqrt{W}(\alpha_N+\alpha_P) \leq \epsilon$.
  Moreover, 
	$\sqrt{W}(\alpha_N+\alpha_P) \leq 2\epsilon$ since $\alpha_P-\alpha_N =\delta_1 \leq \epsilon$
  by~\eqref{eqn:NP}.
	But we also have
	$\alpha_N+\alpha_P\geq \frac{\sqrt{W}}{1+\sqrt{W}}$ by 
  $|w_i -\sqrt{W}| \geq \frac{\sqrt{W}}{1+\sqrt{W}}$.
	It follows that $\frac{W}{1+\sqrt{W}} \leq \sqrt{W}(\alpha_N+\alpha_P) \leq 2\epsilon$, which is a
	contradiction with our choice of~$\epsilon$.
\end{proof}

We now use the above developments to prove the NP-hardness of the reachability and safety problems for MEMDPs.

\begin{proof}[Proof of Theorem~\ref{thm:nphard}]
  Observe that~$W$ can be computed in polynomial time.
  For the  safety problem, note that by the previous lemma,
  the existence of a strategy~$\sigma$ with~$\forall i=1,2, p_i^\sigma \leq \sqrt{W}$ is equivalent to the existence
  of a pure strategy~$\pi$ with~$\forall i=1,2, p_i^\pi = \sqrt{W}$, which we proved to be equivalent to the existence of a solution
  of the subset product problem; so the hardness follows.
  For the reachability problem, we simply note that in our MEMDP~$M$, under any strategy, the sum of the reachability probabilities of~$T$ and~$\bot$
  equals~$1$. Thus, if we write~$q_i^\sigma = \pr_{M_i,s_1}^\sigma[\Diamond \bot]$, we get that for any strategy~$\sigma$, 
  \[
  \forall i=1,2, p_i^\sigma \leq \sqrt{W} \Leftrightarrow 
  \forall i=1,2, q_i^\sigma \geq 1- \sqrt{W}.
  \]
  So the existence of a strategy achieving probabilities at least $(1-\sqrt{W},1-\sqrt{W})$
  is equivalent to the existence of a solution in the subset product problem.
\end{proof}

The hardness of the $\epsilon$-gap problems also follow immediately from the previous lemma.

\begin{theorem}
	\label{lemma:gaphard}
	The $\epsilon$-gap problem for MEMDPs is NP-hard.
\end{theorem}	
\begin{proof}
	We reduce \textsf{Product-Partition} to the $\epsilon$-gap problem for reachability and safety
    in MEMDPs.
    We consider the reduction above, noting that~$\epsilon$ can be computed in polynomial time.
    
    We start by the safety problem, which consists in finding a strategy~$\sigma$ with
    $\forall i=1,2, p_i^\sigma \leq \sqrt{W}$.

	As seen above, if \textsf{Product-Partition} has a solution, then there exists a pure strategy
	in~$M$ with reachability probabilities equal to~$(\sqrt{W},\sqrt{W})$, so the~$\epsilon$-gap instance is
	positive. If \textsf{Product-Partition} has no solution, then there is
	no pure strategy whose reachability probabilities are~$(\sqrt{W},\sqrt{W})$.
    Therefore, by Lemma~\ref{lemma:hardness-epsilon}, there is no
	strategy whose reachability probabilities are component-wise at most $(\sqrt{W}+\epsilon,\sqrt{W}+\epsilon)$. Thus, the $\epsilon$-gap instance is negative.

    For the reachability problem, we similarly consider as target $\bot$, so the question is whether
    for some strategy~$\sigma$, $q_i^\sigma\geq 1-\sqrt{W}$. 
    As in the safety case, if \textsf{Product-Partition} has a solution, then a pure strategy exists
    achieving~$p_i = \sqrt{W}$, which means~$q_i = 1-\sqrt{W}$ for both~$i=1,2$.
    Otherwise, by Lemma~\ref{lemma:hardness-epsilon}, for any~$\sigma$,
    $\exists i=1,2, p_i^\sigma > \sqrt{W}+\epsilon$, which means that
    $\exists i=1,2, q_i^\sigma < 1 - \sqrt{W} - \epsilon$.
\end{proof}

\subsection{Fixed-Memory Strategies}

As an upper bound on the above problem, we show that quantitative reachability for
strategies with a fixed memory size can be solved in polynomial space.
The algorithm consists in encoding the strategy and the probabilities achieved
by each state and each environment, as a bilinear equation, and solving these
in polynomial space in the equation size (see~\cite{Canny-stoc98} for general
polynomial equations).

This case will be used, in the next section, to derive an approximation algorithm for the general problem.

\newcommand\setsyes{\ensuremath{S^{\text{yes}}}}
\newcommand\setsno{\ensuremath{S^{\text{no}}}}
\newcommand\setsrest{\ensuremath{S^{\text{?}}}}

We start by analyzing the case of MDPs.
Given an MDP~$M=(S,A,\delta,r)$, and target set~$T$, consider a subset $\setsno$ of states
and $\setsrest = S \setminus (\setsno \cup T)$. We will write an equation to
solve the reachability problem as follows.
For a starting state~$s_0$, and desired reachability probability~$\lambda$, 
we define the following equation with unknowns $x_s, p_{s,a}$ for all $s
\in \setsrest$, $a \in A(s)$.
\begin{equation}
	\label{eqn:reach}
	\begin{array}{l}
		\forall s \in \setsno, x_s = 0,\\
		\forall s \in T, x_s = 1,\\
		\forall s \in \setsrest, x_s = \sum_{a \in A(s)} p_{s,a} \sum_{t \in S} \delta(s,a,t) x_t,\\
		\forall s \in S, \sum_{a \in A(s)} p_{s,a} = 1,\\
		\forall s \in S, a \in A(s), p_{s,a} \geq 0,\\
		x_{s_0}\geq \lambda
	\end{array}
\end{equation}

For any solution~$(\bar{x},\bar{p})$ of~\eqref{eqn:reach}, let us denote by
$\sigma_{\bar{p}}$ the strategy defined by $\sigma_{\bar{p}}(s,a) =
p_{s,a}$.
Let us also denote by $M^{\bar{p}}$ the Markov chain
obtained from~$M$ by fixing the probability of each action~$a$ from~$s$ to
$p_{s,a}$.
\begin{lemma}
	\label{lemma:reacheqn}
	Consider any~$\setsno\subseteq S$ and any solution 
	$\bar{x},\bar{p}$ of~\eqref{eqn:reach}. If all states~$s$
	of $M^{\bar{p}}$ with zero probability of reaching~$T$ belong to $\setsno$,
	then $x_s = \pr_{M,s}^{\sigma_{\bar{p}}}[\objreach(T)]$.
	Conversely, for any stationary strategy~$\sigma$, 
	such that $\pr_{M,s_0}^\sigma[\objreach(T)]\geq \lambda$,
	there exists a subset~$\setsno\subseteq S$ such that
	$x_s = \pr_{M,s}^\sigma[\objreach(T)]$ and $p_{s,a} = \sigma(s,a)$
	are the unique solution of~\eqref{eqn:reach}.
\end{lemma}
\begin{proof}
	Fix any solution~$(\bar{x},\bar{p})$ of~\eqref{eqn:reach}, and assume that 
	all states~$s$ with a probability of~$0$ of reaching~$T$ satisfy $s \in
	\setsno$. Then~$\bar{x}$ is the solution of the equation obtained by
	fixing~$\bar{p}$.
	But this equation has a unique solution which gives the reachability
	probabilities from each state (see \textit{e.g.} \cite[Theorem
	10.19]{BK-book08}).

	Conversely, given a stationary strategy~$\sigma$, we can define $\setsno$ as
	the set of states from which no path leads to~$T$ in the Markov chain
	$M^{\sigma}$, and by fixing the probabilities
	$p_{s,a} = \sigma(a\mid s)$ in \eqref{eqn:reach}, the unique
	solution is the vector of reachability probabilities.
\end{proof}

We now adapt \eqref{eqn:reach} to MEMDPs and prove the following theorem.

\begin{theorem}
	\label{lemma:fixed-memory-reach}
	The quantitative reachability and safety problems for $K$-memory strategies can
	be solved in polynomial space in~$K$ and in the size of~$M$.
\end{theorem}
\begin{proof}
	We give the proof for reachability objectives. The case of safety is very
	similar and will be sketched.

	For any MEMDP~$M$, and given target states $T$, let us fix
	$\setsno_i\subseteq S$, for each $M_i$. 
	Given~$K$, define the set $M = \{1,\ldots,K\}$ of \emph{memory elements}, and
	fix an initial memory element $m_0 \in M$.
	Given desired reachability probabilities $\alpha_1,\alpha_2$ from state~$s_0$, 
	we write the following equation $E(\setsno_1,\setsno_2)$.
	\begin{equation}
	\label{eqn:reach-memdp}
	\begin{array}{l}
		\forall s \in \setsno_1, m \in M, x_{s,m} = 0,\\
		\forall s \in T, m \in M, x_{s,m} = 1,\\
		\forall s \in \setsrest_1, m \in M, 
		x_{s,m} = \sum_{a \in A(s), m' \in M} \sum_{t \in S} p_{s,m}(a,m')
		\delta_1(s,a,t) x_{t,m'},\\
		\forall s \in \setsno_2, m \in M, y_{s,m} = 0,\\
		\forall s \in T, m \in M, y_{s,m} = 1,\\
		\forall s \in S, m \in M, y_{s,m} = \sum_{a \in A(s), m' \in M} \sum_{t \in S} p_{s,m}(a,m') \delta_2(s,a,t) y_{t,m'},\\
		\forall s \in S, m \in M, \sum_{a \in A(s), m' \in M} p_{s,m}(a,m') = 1,\\
		\forall s \in S, a \in A(s), m,m' \in M, p_{s,m}(a,m') \geq 0,\\
		x_{s_0,m_0}\geq \alpha_1, y_{s_0,m_0}\geq \alpha_2.
	\end{array}
	\end{equation}

	The equation consists in embedding the memory in the MDPs. Each unknown
	$p_{s,m}(a,m')$ corresponds to the probability of choosing
	action~$a$ and changing memory to~$m'$ given state~$s$ and memory~$m$.
	Thus $\sum_{m' \in M}p_{s,m}(a,m')$ is the probability of choosing
	action~$a$ at~$s,m$.
	
	Now polynomial space procedure proceeds as follows. We first guess the sets
	$\setsno_1,\setsno_2$, write the equation $E(\setsno_1,\setsno_2)$, and solve
	it in deterministic polynomial space. We then check, for each~$i=1,2$,
	whether all states~$s$ from which the probability of reaching~$T$ is~$0$ belong
	to~$\setsno_i$. We accept if this is the case, and reject otherwise.

	The correctness follows from Lemma~\ref{lemma:reacheqn}. In fact, if there is
	a stationary strategy achieving probabilities $\alpha_1$ and $\alpha_2$ and~$s_0$,
	then there exist the sets $\setsno_1,\setsno_2$ of
	$0$-probability states, and for this guess \eqref{eqn:reach-memdp} has a
	solution obtained by fixing $p_{s,a} = \sigma(a \mid s)$, and where
	$x_s$ is the probability achived in~$M_1$ at~$s$, and~$y_s$ at~$M_{2}$.
	Therefore the procedure accepts. If there
	is no such strategy, then for all guesses, either the desired probabilities do not 
	satisfy the lower bounds, or one of the sets $\setsno_i$ does not contain all 
	$0$-probability states.

	The problem can be solved similarly for safety properties. In fact, the
	events of avoiding $T$ and reaching~$T$ are complementary.
	Equation \eqref{eqn:reach} and Lemma~\ref{lemma:reacheqn} can be adapted for
	safety objectives by simply requiring $x_{s_0}\leq \lambda$, which means that
	the safety property holds with probability at least $1-\lambda$
  in Equation~\eqref{eqn:reach-memdp}.
\end{proof}

\subsection{Approximation Algorithm}
\label{section:approximation}
We now show that considering finite-memory strategies are hardly restrictive, in
the sense that they can be used to approximately achieve the value.
We also give a memory bound that is sufficient to approximate the value by any
given $\epsilon$.
\begin{theorem}
	\label{lemma:reach-finite-suffices}
	For any MEMDP~$M$ with only trivial DECs, reachability objective~$\Phi$, strategy~$\sigma$,
	and $\epsilon>0$, 
	there exists a $N$-memory strategy $\sigma'$ with
	\(
		\forall i=1,2, \pr_{M_i,s}^{\sigma'}[\Phi] \geq
		\pr_{M_i,s}^\sigma[\Phi] - \epsilon,
	\)
	where~$N=(|S|+|A|)^{\frac{4|S|^3|A|^2}{p^{|S|}\eta^2}\log^3(1/\epsilon)}$,
	with~$p$ the smallest nonzero probability and 
	$\eta = \min\{|\delta_1(s,a,s') - \delta_2(s,a,s')| \mid
	s,a,s' \text{ s.t. } \delta_1(s,a,s') \neq \delta_2(s,a,s')\}$.
\end{theorem}
\begin{proof}[of Lemma~\ref{lemma:reach-finite-suffices}]
	By Definition~\ref{def:tilde} and Corollary~\ref{corollary:Mtilde}, we assume
	that~$M$ has only trivial DECs. Consider an arbitrary strategy~$\sigma$
	for~$M$. Define $\eta = \min\{|\delta_1(s,a,s') - \delta_2(s,a,s')| \mid
	s,a,s' \text{ s.t. } \delta_1(s,a,s') \neq \delta_2(s,a,s')\}$.
	We call the pair $(s,a)$ \emph{distinguishing} if for some~$s'$, $\delta_1(s,a,s')
	\neq \delta_2(s,a,s')$.
	Let us fix~$K = 2\frac{\log(1/\epsilon)}{\eta^2}$. Let~$p$ denote the smallest
	nonzero probability in~$M$ and $q = p^{|S|}$.

	Strategy~$\sigma'$ is defined identically to~$\sigma$ on all histories up to
	length~$L=l|S|$, where $l \geq
	\left(\frac{2|S||A|}{p^{|S|}\eta^2}\right)^2\log^3(1/\epsilon)$.
	Note that~$L$ is exponential, so the memory requirement is doubly exponential.
	Upon arrival to a DEC (thus, trivial and
	absorbing) it switches to a memoryless strategy.
	On any other history~$h_1\ldots h_L$, we distinguish cases:

	Assume there is a distinguishing pair that was seen at least~$K$ times
	in~$h$, and consider~$(s,a)$ the first such pair.
	Let us write $d_i = \delta_i(s,a,s')$ for some~$s'$ with $d_1\neq
	d_2$. Assume $d_i< d_{3-i}$ for some~$i=1,2$. Define~$c_{s,a}^L$ as the random
	variable denoting the number of occurences of $(s,a)$ in a prefix of length~$L$, 
	and $c_{s,a,s'}^L$ the number of times the state~$s'$ was reached after $(s,a)$.
	In~$\sigma'$, if $|\frac{c_{s,a,s'}^L}{c_{s,a}^L} d_i| < \frac{|d_1 - d_2|}{2}$, then we switch to the
	memoryless optimal strategy for~$M_i$.
	If no distinguishing pair satisfies this condition for any~$i=1,2$, then we
	switch to some arbitrary memoryless strategy.
	Strategy~$\sigma'$ is clearly finite-memory using $(|S|\cdot|A|)^L$ memory
	elements, for any choice of~$l$.

	First, let us show that conditioned on the event that some distinguishing pair
	was observed~$K$ times, the strategy~$\sigma'$ is $\epsilon$-optimal. In fact,
	By Hoeffding's inequality, for an edge $(s,a,s')$, we have
	for any strategy~$\tau$, 
	\[
		\pr_{M_i,s}^{\tau}\left[\left\vert\frac{c_{s,a,s'}^L}{c_{s,a}^L} -
		d_i\right\vert \geq \left|\frac{d_2-d_1}{2}\right| \mid
	c_{s,a}^L \geq K\right]
		\leq e^{-2K\frac{d_2-d_1}{2}^2} \leq \epsilon,
	\]
	which means that $\sigma'$ will switch to the optimal strategy for~$M_i$ from
	with probability at least $1-\epsilon$. 

	Let us denote by~$T_K^L$ the event that $\exists (s,a), c_{s,a}^L=K$, and $D^L$ the
	event that some DEC (therefore, trivial and absorbing) is reached.
	The rest of the proof consists in
	showing that with high probability either~$T_K^L$ occurs or the play is stuck in some absorbing
	state, and in any such history $\sigma'$ performs as good as~$\sigma$ up to~$\epsilon$.

	\medskip
	\textbf{Either $T_K^L$ or a DEC.}
	We will show that either $T_K^L$ or $D^L$ occurs with probability
	$1-\epsilon$.
	
	Let~$X_j$ denote the random variable giving the state at~$j$-th step,
	and~$A_j$ the $j$-th action.
	For any history $h$, let~$y(h)$ denote the number of states of~$h$ belonging to a DEC
	+ the number of distinguishing actions in~$h$.
	We show that, under any strategy~$\tau$, and for any state~$s$, $i=1,2$,
	$\pr_{M_i,s}^\tau[y(X_1A_1\ldots A_{|S|-1}X_{|S|})\geq 1]\geq q$.
	To prove this, we first write $\tau$ as a linear combination of strategies
	that are deterministic in the first $|S|$ steps: $\tau = \sum_{i} \lambda_i
	\pi_i$ for $(\pi_i)_i$ a finite family of strategies that are pure in the
	first $|S|$ steps, and $(\lambda_i)_i$ such that $\sum_i \lambda_i = 1$.
	We have that 
	$\pr_{M_i,s}^\tau[y(X_1A_1\ldots A_{|S|-1}X_{|S|})]
	= \sum_j \lambda_j \pr_{M_i,s}^{\pi_j} [y(X_1A_1\ldots A_{|S|-1}X_{|S|})\geq 1]$.
	We will prove that for each~$\pi_j$, 
	\[
		\pr_{M_i,s}^{\pi_j} [y(X_1A_1\ldots
		A_{|S|-1}X_{|S|})\geq 1]\geq q
	\].
	We consider the unfolding of depth $|S|$ from
	state~$s$ under strategy~$\pi_j$. If this unfolding contains a state~$t$ of a DEC,
	then the path from~$s$ to~$t$ has probability at least~$q$ under
	strategy~$\pi_j$ since it is deterministic in the first $|S|$ steps, 
	and the result follows. 
	If the unfolding contains a distinguishing action, then it will be taken
	similarly with probability at least~$q$.  
	Otherwise, assume the unfolding contains no DEC or distinguishing action.
	But in this case, if we cut each branch whenever a state is visited twice, we
	obtain an end-component in~$M_i$. Since no action is distinguishing, this is a
	non-distinguishing double end-component, which is a contradiction.

	It follows that 
	$\expect_{M_i,s}^\tau[y(X_1A_1\ldots A_{|S|-1}X_{|S|})]\geq q$ for any
	state~$s$ and any strategy~$\tau$.
	We factorize a given history of length~$L$ in to factors of length~$|S|$. Let~$Y_j$ be the random
	variable denoting $y(h_{(j-1)|S|+1\ldots j|S|})$. We just showed that
	$\expect_{M_i,s}^\tau[Y_j]\geq q$ for any strategy~$\tau$, state~$s$ and
	$j=1\ldots l$.
	Let $Y=\sum_{j=1}^l Y_j$.
	We use Hoeffding's inequality to write
	\[
		\pr_{M_i,s}^\tau[Y \leq \expect[Y] - t] \leq e^{-\frac{t^2}{2l|S|^2}},
	\]
	for any~$t>0$, since $\left|\expect_{M_i,s}^\tau[Y_j]\right|\leq 2|S|$. We get that 
	$\pr_{M_i,s}^\tau[Y \leq lq - t] \leq e^{-\frac{t^2}{2l|S|^2}}$ since $lq \leq
	\expect[Y]$. We would like to obtain that
	$\pr_{M_i,s}^\tau[Y \leq |S|\cdot |A|\cdot K ] \leq \epsilon$, which means
	that with probability at least~$1-\epsilon$, either $T_K^L$ or $D^L$ holds.
	Therefore, in the above equation, we require $e^{-\frac{t^2}{2l|S|^2}} \leq
	\epsilon$, which means
	\begin{equation}
		\label{eqn:t2l}
		\frac{t^2}{l} \geq 2\log(1/\epsilon)|S|^2,
	\end{equation}
	and we let $t = lq - |S||A|K$. To get~\eqref{eqn:t2l}, it suffices to ensure
	$\frac{(lq-|S||A|K)^2}{l} \geq 2\log(1/\epsilon)|S|^2$, which holds for
	our choice of~$l$.

	\medskip
	\textbf{End of the proof}
	We write $\pr_{M_i,s}^{\sigma'}[\phi]$ as
	\[\begin{array}{l}
			\pr_{M_i,s}^{\sigma'}[\phi \mid T_K^L \lor D^L ]
			\pr_{M_i,s}^{\sigma'}[T_K^L \lor D^L]
			+ \pr_{M_i,s}^{\sigma'}[\phi \mid \lnot T_K^L \land \lnot D^L]
			\pr_{M_i,s}^{\sigma'}[\lnot T_K^L \land \lnot  D^L]\\
		\end{array}
	\]
	We clearly have $\pr_{M_i,s}^\sigma[T_K^L\lor
	D^L]=\pr_{M_i,s}^{\sigma'}[T_K^L\lor D^L]$,
	and we showed above that
	$\pr_{M_i,s}^{\sigma'}[T_K^L \lor D^L]\geq 1-\epsilon$. 
	Thus, using the same decomposition, $\pr_{M_i,s}^\sigma[\phi \mid T_K^L \lor D^L]
	\geq \pr_{M_i,s}^\sigma[\phi] - \epsilon$.

	We will show that $\pr_{M_i,s}^{\sigma'}[\phi \land (T_K^L \lor D^L)]
	\geq (1-\epsilon) \pr_{M_i,s}^{\sigma}[\phi \land (T_K^L \lor D^L)]$,
	which implies $\pr_{M_i,s}^{\sigma'}[\phi \mid T_K^L \lor D^L ]
	\geq (1-\epsilon) \pr_{M_i,s}^{\sigma}[\phi \mid T_K^L \lor D^L]$.
	But let us first show how we conclude. Because
	$\pr_{M_i,s}^{\sigma'}[\phi]\geq \pr_{M_i,s}^{\sigma'}[\phi \mid
	T_K^L \lor D^L] \pr_{M_i,s}^{\sigma'}[T_K^L \lor D^L]$, combining with the
	above inequality, it follows
	\[ 
		\begin{array}{ll}
			\pr_{M_i,s}^{\sigma'}[\phi] &\geq (1-\epsilon)^2 (\pr_{M_i,s}^{\sigma}[\phi]
		- \epsilon)\\
		&\geq \pr_{M_i,s}^{\sigma}[\phi] - 3\epsilon,
		\end{array}
	\]
	as desired.
		
	We write 
	$\pr_{M_i,s}^{\sigma'}[\phi \land (T_K^L \lor D^L)]
	=\pr_{M_i,s}^{\sigma'}[\phi \land T_K^L]
	+ \pr_{M_i,s}^{\sigma'}[\phi \land D^L \land \lnot T_K^L]$.
	We have $\pr_{M_i,s}^{\sigma'}[\phi \mid D^L \land \lnot T_K^L]= 
	\pr_{M_i,s}^{\sigma}[\phi \mid D^L \land \lnot T_K^L]$
	for both~$i=1,2$ since the history ends in an absorbing state.
	It follows that 
	$\pr_{M_i,s}^{\sigma'}[\phi \land D^L \land \lnot T_K^L]= 
	\pr_{M_i,s}^{\sigma}[\phi \land D^L \land \lnot T_K^L]$.
	We now show that 
	$\pr_{M_i,s}^{\sigma'}[\phi \mid T_K^L ]
	\geq (1-\epsilon) \pr_{M_i,s}^{\sigma}[\phi \mid T_K^L ]$
	which implies similarly 
	$\pr_{M_i,s}^{\sigma'}[\phi \land T_K^L ]
	\geq (1-\epsilon) \pr_{M_i,s}^{\sigma}[\phi \land T_K^L ]$
	since $\pr_{M_i,s}^{\sigma'}[T_K^L] = \pr_{M_i,s}^{\sigma}[T_K^L]$.
  Let $S_K^L(i)$ denote
	the event that for the first distinguishing pair $(s,a)$  that appears~$K$
	times in the prefix of length~$L$, $|\frac{c_{s,a,s'}^L}{c_{s,a}^L} - d_i|
	\leq |\frac{d_1-d_2}{2}|$.
	We have
	\[
		\begin{array}{ll}
		\pr_{M_i,s}^{\sigma'}[\phi \mid T_K^L]
		&=
		\pr_{M_i,s}^{\sigma'}[\phi \mid S_K^L(i) \land T_K^L ]
		\pr_{M_i,s}^{\sigma'}[S_K^L(i) \mid T_K^L]
		+\\&~
		\pr_{M_i,s}^{\sigma'}[\phi \mid \lnot S_K^L(i) \land T_K^L]
		\pr_{M_i,s}^{\sigma'}[\lnot S_K^L(i) \mid T_K^L].
		\end{array}
	\]
	For any~$i=1,2$, we have $\pr_{M_i,s}^{\sigma'}[\lnot S_K^L(i) \mid T_K^L]\leq \epsilon$ as we
	showed above,
	so $\pr_{M_i,s}^{\sigma'}[S_K^L(i) \mid T_K^L]\geq 1-\epsilon$, and 
	$\pr_{M_i,s}^{\sigma'}[\phi \mid S_K^L(i) \land T_K^L]
	\geq \pr_{M_i,s}^{\sigma}[\phi \mid S_L^L(i) \land T_K^L ]$
	since~$\sigma'$ switches to the optimal strategy for~$M_i$.
	This shows that 
		$\pr_{M_i,s}^{\sigma'}[\phi \mid T_K^L ]\geq
		(1-\epsilon)\pr_{M_i,s}^{\sigma'}[\phi \mid T_K^L ]$.
\end{proof}

Combining Theorems~\ref{lemma:fixed-memory-reach}
and~\ref{lemma:reach-finite-suffices},
we derive an approximation algorithm:

\begin{theorem}
	\label{lemma:gapalgorithm}
	There is a procedure that works in $O(N\cdot|M|)$ space solving the~$\epsilon$-gap problem
	for quantitative reachability in MEMDPs.
	Moreover, whenever the procedure answers YES, there exists a strategy~$\sigma$
	such that $\forall i=1,2, \pr_{M_i,s}^{\sigma}[\objreach(T)]\geq \alpha_i-\epsilon$.
\end{theorem}
\begin{proof}[of Lemma~\ref{lemma:gapalgorithm}]
	We compute~$N$ given by Lemma~\ref{lemma:reach-finite-suffices}, which is doubly exponential in input,
	and apply Lemma~\ref{lemma:fixed-memory-reach} for $N$-memory
	strategies and target probabilities $\alpha_1-\epsilon$ and~$\alpha_2-\epsilon$.
	We solve the equation in polynomial space in the equation size,
	and answer yes if, and only if there is a solution.

	By Lemma~\ref{lemma:reach-finite-suffices}, if there exists a strategy
	achieving $(\alpha_1,\alpha_2)$, there is a $N$-memory strategy achieving
	$(\alpha_1-\epsilon,\alpha_2-\epsilon)$. So the procedure will answer yes.
	If no strategy achieves $(\alpha_1-\epsilon,\alpha_2-\epsilon)$, then, in
	particular, no finite-memory strategy achieves this vector, and the procedure
	will answer no.

	Last, observe that whenever the procedure answers yes, there exists a
	finite-memory strategy achieving $(\alpha_1-\epsilon,\alpha_2-\epsilon)$.
\end{proof}

In our case, the ``gap'' can be chosen arbitrarily small, and the procedure is
used to distinguish instances that are clearly feasible from those that are
clearly not feasible, while giving no guarantee in the borderline. Notice that
we do not have false positives; when the procedure answers
positively, the probabilities are achieved up to~$\epsilon$.

%

\section{Safety and Parity Objectives} \label{section:parity}
\subsection{The Almost-sure Case}
We consider safety and parity objectives, building on techniques
developed for reachability.
Recall that almost-sure and sure safety coincide in MEMDPs.
The equivalence of these with limit-sure safety is less
trivial, and follows from Lemma~\ref{lemma:only11-inf}:
\begin{lemma}
	Limit-sure safety is equivalent to almost-sure safety in MEMDPs.
\end{lemma}
Under Assumption~\ref{assum:absorbing}, safety is a special case of 
parity objectives; we rely on algorithms for parity to decide almost-sure safety
objectives. For quantitative safety, the results of the previous section can be adapted without
difficulty, but we omit the details.

Our first result is a polynomial-time algorithm for almost sure parity objectives.

By Lemma~\ref{lemma:mdp-mec-parity}, we know that for MDPs with parity objectives, inside end-components, the value
of a Parity objective is either~$0$ or~$1$, and in the latter case
a memoryless strategy which only depends on the support of the distributions
exists.
Thus, let us call an end-component~$D$ \emph{$\Phi$-winning} if there
exists a strategy inside~$D$ that satisfies~$\Phi$.

We denote by~$R^\Phi$ the set of revealed states from which~$\Phi$ surely holds.
\begin{algorithm}
	\KwData{MEMDP~$M$, $s \in S$, parity objective $\Phi$}
	 $U := \big(\almostsure(M_1,\Phi) \cap \almostsure(M_2,\Phi)\big) \cup R^{\Phi}$\;
 $M'$ := Sub-MEMDP of~$M$ induced by states~$s'$ s.t.  $\Val_{\Wsafety(U)}^*(\cup M,s')= 1$\;
 $T_i$:=  Set of states of $\Phi$-winning MECs of~$M_i'$\;
 \eIf{$\exists \sigma, \forall i=1,2, \pr_{M_i',s}^\sigma[\objreach(T_1 \cup T_2)]=1$}{
	 $\sigma'$ := Modify $\sigma$ as follows.
					At any state $s \in T_1$ (resp. $s \in T_2\setminus T_1$),
					if~$D$ denotes the MEC of~$M_1$ (resp. $M_2$) which
					contains~$s$, switch to a memoryless strategy winning for~$\Phi$
					compatible with~$D$\;
	 Return $\sigma'$\;
 }{
	Return NO\;
 }
 \caption{Almost-sure parity algorithm for MEMDPs}
	\label{alg:asparity}
\end{algorithm}

\begin{lemma}
	For any MEMDP~$M$, state~$s$, and parity objective~$\Phi$,
	Algorithm~\ref{alg:asparity} decides whether
	there exists a strategy achieving $\Phi$
	almost surely in~$M$, and computes a witnessing memoryless strategy.
\end{lemma}
\begin{proof}
		Consider an instance~$M$, $s$ and~$\Phi$ for which the algorithm answers
		positively. 
%
		Under the returned strategy~$\sigma'$, some state~$s'$ from $T_1\cup T_2$ is visited almost
		surely for the first time in~$M_i$.
		If~$\sigma'$ switches to an optimal strategy for~$M_i'$ (that is, either $i=1$
		and~$s' \in T_1$, or $i=2$ and $s' \in T_2\setminus T_1$), then~$\Phi$ holds
		almost surely in~$M_i$ by definition.
		Otherwise, $s' \in D$ for a $\Phi$-winning MEC~$D$ of~$M_{3-i}'$ and~$\sigma'$ switches to an optimal
		strategy for~$M_{3-i}'$ that stays in~$D$.
		Let~$D_1',\ldots,D_m'$ be the set of all end-components included in~$D$ such that
		$\pr_{M_{3-i}',s'}^{\sigma'}[\Inf=D_j']>0$.
		First, observe that by Assumption~\ref{assum:revealed}, $R^\Phi \cup \bigcup_{j=1}^m D_j'$ is reached almost surely in~$M_i'$
		under~$\sigma'$ from~$s'$, since~$\sigma'$ is compatible with~$D$
		in~$M_{3-i}'$. Moreover, $\pr_{M_i',s'}^{\sigma'}[\Inf \in
		\{D_1',\ldots,D_m'\} \cup R^\Phi ]=1$.
		If some~$D_j'$ is a DEC, then $\sigma'$ also almost surely satisfies~$\Phi$
		in~$M_i'$ from~$D_j'$ by Lemma~\ref{lemma:mdp-mec-parity}. Otherwise, some action~$a$ from some state~$t$
		of~$D_j'$ has a different support in~$M_i'$ and~$M_{3-i}'$. Because~$D_j'$ is an
		end-component for~$M_{3-i}'$, we have $\supp(\delta_{3-i}(t,a)) \subsetneq
		\supp(\delta_i(t,a))$ since otherwise $D_j'$ would contain an absorbing state
		different than~$t$ (by Assumption~\ref{assum:revealed}), which is in contradiction with the fact that it is an
		end-component in~$M_{3-i}'$. Therefore, starting at~$D_j'$ in $M_i'$, under~$\sigma'$, the play
		almost surely leaves~$D_j'$ for a $i$-revealed state. By definition of~$M'$
		such a state is in~$R^\Phi$. Thus, $\pr_{M_i,s'}^{\sigma'}[\Phi]=1$.

		Conversely, assume that there exists~$\tau$ such that 
		$\pr_{M_i,s}^{\tau}[\Phi]=1$ for all~$i=1,2$. Observe that
		$\pr_{M_i,s}^{\tau}[\Wsafety(U)]=1$ since otherwise for some~$i=1,2$, we
		reach a state that is not almost surely winning. 
		Strategy~$\tau$ is therefore compatible with~$M_i'$ and~$s \in M_i'$.
		Recall that a strategy satisfies a parity condition almost surely in an MDP
		if, and only if the set of winning MECs is reached almost surely.
		Hence, we must have $\forall i=1,2, \pr_{M_i',s}^\tau[\objreach(T_i)]=1$, 
		so in particular $\forall i=1,2, \pr_{M_i',s}^\tau[\objreach(T_1\cup T_2)]=1$,
		and the algorithm answers positively.
\end{proof}

This yields the following theorem.
\begin{theorem}
\label{thm:asparity}
	The almost-sure parity problem is decidable in polynomial time.
\end{theorem}

\subsection{The Quantitative Case: Reduction to Reachability}
Our second result is a polynomial-time reduction from the quantitative parity
problem
to the quantitative reachability problem which preserves value vectors.
It follows 1) a
polynomial-time algorithm for the limit-sure parity problem, 2) and that any algorithm for
solving the quantitative reachability problem can be used to solve
the quantitative parity problem. 
In particular, results of Section~\ref{section:quantitative} applies to parity
objectives.

The idea of the reduction is similar to previous constructions. We
modify~$\hat{M}$ by adding new transitions from each MEC~$D$ of each~$M_i$ to
fresh absorbing states with probability equal to the probability of winning
from~$D$ in~$M_i$.

\begin{definition}
	Given a MEMDP~$M$,
	we define $\bar{M} = (\bar{S}, \bar{A}, \bar{\delta}_1, \bar{\delta_2})$
	by modifying~$\hat{M}$ as follows.
	For any~$i=1,2$, non-trivial MEC~$D$ of~$\hat{M}_i$, and state~$s \in D$, we add an action
	$a_D$ from~$s$. In~$\bar{M}_i$, $a_D$ leads to a fresh absorbing state $t_D^0$
	with even parity 
	with probability $\Val^*_\calP(M_i,s)$, and to a fresh absorbing 
	state~$t_D^1$ with odd parity with remaining probability.
	In~$\bar{M}_{3-i}$, it leads to a losing absorbing state~$t_D^1$.
	Let $\bar{\Phi}$ be the reachability
	objective with targets the absorbing states
	of~$\cup \bar{M}$ with even parity.
\end{definition}
Observe that the set of absorbing states of~$\cup \bar{M}$ with even parity 
is exactly $\{W_D\mid D\text{ DEC of}~\hat{M}\} \cup \{t_D^0 \mid \exists j=1,2, D \in
\calG(\bar{M}_j)\}$. 
For any state~$s$ of~$M$, let ${\barabstr}(s)$ denote the state in~$\bar{M}$ to which
it is mapped by our construction: for any $s$ belonging to a MDEC~$D$,
$\barabstr(s)=s_D$, and ${\barabstr}(s) = s$ otherwise.
Note that~$\bar{M}$ can be constructed in polynomial time since MECs can be
computed in polynomial time.

We will prove that achieving a pair of satisfaction probabilities for a parity
objective~$\calP_p$ in~$M$ is equivalent to achieving the same probabilities for the
reachability objective~$\bar{\Phi}$ in~$\bar{M}$.

We start with two simple technical lemmas.
Let us denote by $\calG(\bar{M}_j)$ the set of MECs of~$\bar{M}_j$ that are not
DECs. The following lemma gives a classification of the MECs of~$M_i$ with
respect to~$\bar{M}_i$.
\begin{lemma}
	\label{lemma:wtmec}
	Let~$D$ be an end-component of~$M_i$ which is not a DEC.
	Then, either~$D$ contains a distinguishing DEC, or
	$D \subseteq \barabstr^{-1}(E)$ for some~$E \in \calG(\bar{M}_i)$.
\end{lemma}
\begin{proof}
	Assume that~$D$ is not a DEC and does not contain distinguishing DECs. Notice
	that~$D$ might contain non-distinguishing DECs.
	By construction $\barabstr(D)$ is
	an end-component in~$\bar{M}_i$: it is $\delta_i$-closed and strongly connected. 
	We have $D\subseteq\barabstr^{-1}(\barabstr(D))$.
\end{proof}

The following lemma is an adaptation of Lemma~\ref{lemma:mec-transient-K} to
paths in~$M_{3-j}$ that stay in the preimage of a MEC of~$\bar{M}_j$.
\begin{lemma}
	\label{lemma:wtmec-K}
	For any MEMDP $M$, and $\epsilon>0$, there exists~$K$ such that for any $j=1,2$, $D \in
	\calG(\hat{M}_j)$, and any history $h \in \calH(\bar{M})$ which contains a
	factor of length $K$ compatible with~$D$, $\pr_{M_{3-j},s}^\tau[\red^{-1}(h)]\leq \epsilon$ for any
	strategy~$\tau$.
\end{lemma}
\begin{proof}
	For any~$\epsilon>0$, fix~$K$ as in Lemma~\ref{lemma:mec-transient-K}. Fix any $D \in \calG(\hat{M}_j)$, and
	history~$h \in \calH(\hat{M})$ of length $K$ compatible
	with~$D$. 
	We know that $\pr_{\hat{M}_{3-j},h_1}^\tau[h]\leq \epsilon$ by
	Lemma~\ref{lemma:mec-transient-K}.
	History~$h$ does not contain states~$\calT$ since otherwise it enters an
	absorbing state.
	It follows, from Lemma~\ref{lemma:wtsigma} that
	$\pr_{M_{3-j},s}^\tau[\red^{-1}(h)]\leq \epsilon$.
\end{proof}

We can now prove the following direction.
\begin{lemma}
	\label{lemma:wtM-larger}
	Consider any MEMDP~$M$, and parity condition~$\calP$. 
	For any state~$s$, strategy~$\sigma$, and~$\epsilon>0$,
	there exists a strategy~$\bar{\sigma}$ such that
	$\pr_{\bar{M}_i,\bar{s}}^{\bar{\sigma}}[\bar{\Phi}]\geq
	\pr_{M_i,s}^\sigma[\calP] - \epsilon$.  \end{lemma}
\begin{proof}
	Let~$\bar{\sigma}$ as defined in Lemma~\ref{lemma:wtsigma}.
	We define $\bar{\sigma}'$ from~$\bar{\sigma}$ as follows. 
	For any~$\epsilon$ let~$K$ be as defined in Lemma~\ref{lemma:wtmec-K}.
	For any~$j=1,2$, and $D \in \calG(\bar{M}_j)$, define $\calD_K^j(D)$
	as the set of histories in~$\calH_\calT(\bar{M})$ whose suffix of length $K$ is compatible with~$D$,
	and such that no proper suffix contains a factor of length~$K$ compatible with
	any $D' \in \calG(\bar{M}_1)\cup \calG(\bar{M}_2)$.
	Let~$\calD_K^j$ denote the union of all $\calD_K^j(D)$,
	and $\calD_K = \calD_K^1\cup \calD_K^2$.

	The following events are disjoint and occur almost surely in each~$\bar{M}_j$
	under~$\bar{\sigma}'$:
	\begin{enumerate}
		\item[$\bar{E}_1$]: $\calD_K{\bar{S}}^\omega$.
		\item[$\bar{E}_2$]: $\calH_\calT \calT {\bar{S}}^\omega \setminus \bar{E}_1$.
		\item[$\bar{E}_3$]: $\calH_\calT \cdot a_D^\$\cdot {\bar{S}}^\omega \setminus \bar{E}_1$ for some non-distinguishing DEC~$D$.
	\end{enumerate}
	This follows from the fact that states and actions seen infinitely often in
	$\bar{M}_j$ under $\bar{\sigma}'$ is almost surely an end-component. If such an
	end-component is in $\calG(\bar{M}_j)$ then~$\bar{E}_1$ holds. If $\bar{E}_1$
	does not hold, then any such end-component is a DEC, thus a trivial DEC.
	Because the DEC~$\{t_D\}$ is only reachable if~$\calD_K$ occurs by definition
	of~$\bar{\sigma}'$, any such end-component corresponds to a distinguishing or
	non-distinguishing DEC.

	Similarly, the following events are disjoint and occur almost surely in each $M_j$ under~$\sigma$:
	\begin{enumerate}
		\item[$E_1$]: $\red^{-1}(\calD_K)S^\omega$.
		\item[$E_2$]: $\red^{-1}(\calH_\calT)\barabstr^{-1}(\calT) \setminus E_1$.
		\item[$E_3$]: $\red^{-1}(\calH_\calT)D^\omega \setminus E_1$ for some non-distinguishing DEC~$D$.
	\end{enumerate}
	In fact, if the play stays in
	an end-component that belongs to $\barabstr^{-1}(E)$ for some $E \in \calG(\bar{M}_j)$ 
	then we are in~$E_1$. If~$E_1$ is false, then such an end-component either
	contains a distinguishing DEC, in which case $E_2$ holds almost surely, or it
	is a DEC, in which case either~$E_2$ or~$E_3$ holds almost surely.
	
	By Lemma~\ref{lemma:wtmec-K}, we have $\pr_{M_j,s}^\sigma[\red^{-1}(\calD_K(D))] \leq \epsilon$
	for $D \in \calG(\bar{M}_{3-j})$.
	It follows that, for any~$j=1,2$, 
	\[
	\begin{array}{ll}
	\pr_{M_j,s}^\sigma[\calP] &\geq 
		\sum_{H \in \red^{-1}(\calH_\calT\cdot \calT)\setminus E_1} \pr_{{M}_j,{s}}^{{\sigma}}
		[\calP \mid \red^{-1}(H))] \pr_{M_j,s}^\sigma[\red^{-1}(H)]

		\\&~+\sum_{H \in \red^{-1}(\calH_\calT) \setminus E_1, D~\text{non-dist.}}
		\pr_{M_j,s}^\sigma[\calP \mid H\cdot D^\omega] \pr_{M_j,s}^\sigma[H\cdot D^\omega]
		\\&~+\sum_{H \in \red^{-1}(\calD_K^j)S^\omega }
		\pr_{M_j,s}^\sigma[\calP \mid H]
				\pr_{M_j,s}^\sigma[H].
				\\&\geq \pr_{M_j,s}^\sigma[\calP]- \epsilon.
	\end{array}
	\]
	We similarly write
	\[
	\begin{array}{ll}
	\pr_{\bar{M}_j,\bar{s}}^{\bar{\sigma}'}[\bar{\phi}_j] &\geq 
		\sum_{H \in \calH_\calT\cdot \calT \setminus \bar{E}_1} \pr_{\bar{M}_j,\bar{s}}^{\bar{\sigma}'}
				[\bar{\phi}_j \mid H]
		\pr_{\bar{M}_j,\bar{s}}^{\bar{\sigma}'}[H]
		\\&~+\sum_{H \in \calH_\calT\setminus \bar{E}_1, D~\text{non-dist.}}
		\pr_{\bar{M}_j,\bar{s}}^{\bar{\sigma}'}[\bar{\phi}_j \mid H \cdot a_D^\$ (s_D^\$)^\omega]
		\pr_{\bar{M}_j,\bar{s}}^{\bar{\sigma}'}[H \cdot a_D^\$ (s_D^\$)^\omega]
		\\&~+\sum_{H \in \calD_K^j}
		\pr_{\bar{M}_j,\bar{s}}^{\bar{\sigma}'}[\bar{\phi}_j \mid H ]
				\pr_{M_j,s}^\sigma[H].
	\end{array}
	\]

	We will now compare the probability of winning in $M$ and in~$\bar{M}$
	conditioned on events $E_i$ and $\bar{E}_i$.
	First, note that by~\eqref{eqn:wtinduct}, and the definition of~$\bar{\sigma}$,
	we have that for any $H \in \calH_\calT \cdot \calT \cup \calH_\calT a_D^\$
	(s_D^\$)^\omega \cup \calD_K^j$,  $\pr_{M_j,s}^\sigma[\red^{-1}(H)] =
	\pr_{\bar{M}_j,\bar{s}}^{\bar{\sigma}'}[H]$.
	Let us show that conditioned on each of
	these events $\bar{M}_j$ achieves a higher or equal probability under~$\bar{\sigma}$.
	For histories $\calH_\calT \cdot \calT$ this is clear since~$\bar{M}_j$ then
	reaches~$W_D$ with the optimal probability of winnig in~$M_j$.
	For histories $\calH_\calT a_D^\$(s_D)^\omega$, the probability achieved
	in~$\bar{M}_j$ is exactly the probability of winning in~$M_j$ while staying
	inside~$D$, so at least $\pr_{M_j,s}^\sigma[\calP\mid H\cdot D^\omega]$.
  Last, from histories $\calD_K^j(D)$ with $D \in \calG(\bar{M}_j)$, we reach
	in~$\bar{M}_j$ with optimal probability of winning from a corresponding state
	in~$M_j$, so at least $\pr_{M_j,s}^\sigma[\calP\mid h ]$ for any $h \in
	\red^{-1}(H)$.
	It follows that $\pr_{\bar{M}_j,\bar{s}}^{\bar{\sigma}}[\bar{\phi}_j] \geq
	\pr_{M_j,s}^\sigma[\calP] - \epsilon$.
\end{proof}

To prove the converse, we need some additional lemmas.

\begin{lemma}
	\label{lemma:strat-nddec}
	Let~$D$ be a non-distinguishing DEC, and~$\bar{\sigma}$ any strategy
	in~$\bar{M}$. There exists a strategy~$\tau$ such that
	for any history $hs_D (f_1 a_1) s_D \ldots (f_m a_m) s'$, 
	where $(f_ia_i)$ is a pair of frontier state-action, and~$s'$ is a state
	outside of~$D$,
	\[
		\begin{array}{l}
		\pr_{M_j,s_0}^\tau[\red^{-1}(s_D(f_1 a_1) s_D\ldots (f_m a_m) s') \mid h]
		=\pr_{\bar{M}_j,\bar{s_0}}^{\bar{\sigma}}[s_D(f_1 a_1) s_D\ldots (f_m a_m) s'
		\mid h],
		\\
		\pr_{M_j,s_0}^\tau[\red^{-1}(s_D(f_1 a_1) s_D\ldots (f_m a_m) )D^\omega \mid h]
		=\pr_{\bar{M}_j,\bar{s_0}}^{\bar{\sigma}}[s_D(f_1 a_1) s_D\ldots (f_m a_m)
			a_D^\$
		\mid h].
	\end{array}
	\]
\end{lemma}
\begin{proof}
	Consider any history $\red^{-1}(s_D(f_1a_1)\ldots (f_{i-1}a_{i-1})s_D)$.
	let $p_{f,a}$ denote the probability that the pair of frontier state-action
	$(fa)$ is taken for the first time under~$\bar{\sigma}$ from history
	$s_D (f_1a_1) \ldots (f_{i-1}a_{i-1})s_D$.
	We define~$\tau$, by first choosing each pair $(fa)$ with probability
	$p_{f,a}$, and then running a memoryless strategy that reaches state $f$ almost surely, and
	once $f$ is reached chooses~$a$.
	With probability $1- \sum_{f,a} p_{f,a}$, we run any strategy compatible
	with~$D$. This is clearly the probability of~$\bar{\sigma}$ of taking action $a_D^\$$.
\end{proof}
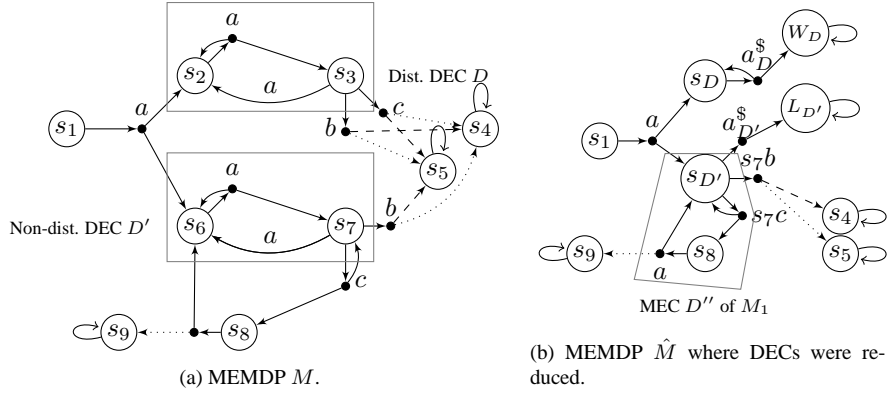
\begin{figure}[h]
	\centering
	\begin{subfigure}[b]{0.6\textwidth}
		\centering
		\begin{tikzpicture}
			\tikzstyle{every state}=[node distance=1cm,minimum size=10pt, inner sep=1pt];
			\node[state] at (0,0) (A){$s_1$};
			\node[state,right of=A,fill=black,minimum size=3pt] (Aa) {};
			\node[state,above right of=Aa] (B){$s_2$};
			\node[state,node distance=0.7cm, above right of=B, fill=black,minimum size=3pt] (Bb) {};
			\node[state,node distance=2cm, right of=B] (C) {$s_3$};
			\node[state,fill=black,minimum size=3pt] at ($(C.south)+(0,-0.5)$) (Cb) {};
			\node[state,fill=black,minimum size=3pt] at ($(C.south)+(+0.5,-0.25)$)(Cc) {};
			\node[state,node distance=4.5cm, right of=Aa] (Top) {$s_4$};
			\node[state,node distance=0.8cm, below left of=Top] (Bot) {$s_5$};

			\node[state,node distance=2cm, below of=B] (D){$s_6$};
			\node[state,node distance=0.7cm, above right of=D, fill=black,minimum size=3pt] (Dd) {};
			\node[state,node distance=2cm, right of=D] (E) {$s_7$};
			\node[state,node distance=0.6cm, right of=E, fill=black, minimum size=3pt] (Eb) {};
			\node[state,node distance=0.8cm, below of=E, fill=black, minimum size=3pt] (Ec) {};
			\node[state,node distance=2cm, below left of=E] (F) {$s_8$};
			\node[state,node distance=0.6cm, left of=F, fill=black, minimum size=3pt] (Ff) {};
			\node[state,node distance=1cm, left of=Ff] (G) {$s_9$};

			\node at ($(Aa.north)+(0,0.2)$) {$a$};
			\node at ($(Bb.north)+(0,0.2)$) {$a$};
			\node at ($(Dd.north)+(0,0.2)$) {$a$};
			\node at ($(Cb.north)+(-0.2,0)$) {$b$};
			\node at ($(Cc.north)+(0.2,0)$) {$c$};
			\node at ($(Eb.north)+(0,0.2)$) {$b$};
			\node at ($(Ec.north)+(0.2,0)$) {$c$};
			\path[-latex']
			(A) edge (Aa)
			(Aa) edge (B)
			(B) edge (Bb)
			(Bb) edge[bend right] (B)
			(Bb) edge (C)
			(C) edge[bend left] node[above]{$a$} (B)
			(C) edge (Cb)
			(C) edge (Cc)
			(Cc) edge[dotted] (Top)
			(Cc) edge[dashed] (Bot)
			(Cb) edge[dashed] (Top)
			(Cb) edge[dotted] (Bot)
			(Aa) edge (D)
			(D) edge (Dd)
			(Dd) edge[bend right] (D)
			(Dd) edge (E)
			(E) edge[bend left] node[above]{$a$} (D)
			(E) edge (Eb)
			(E) edge (Ec)
			(E) edge[bend left] (D)
			(Eb) edge[dashed] (Bot)
			(Eb) edge[dotted,bend right] (Top)
			(Ec) edge (F)
			(Ec) edge[bend right] (E)
			(F) edge (Ff)
			(Top) edge[loop above] (Top)
			(Bot) edge[loop above] (Bot)
			(Ff) edge (D) 
			(Ff) edge[dotted] (G)
			(G) edge[loop left] (G);
			\draw[-,gray,node distance=1cm] ($(B.north west)+(-0.2,0.8)$) rectangle
				($(C.south east)+(0.2,-0.3)$);
			\node[anchor=west] at ($(C.east)+(0.2,0)$) {\scriptsize Dist. DEC~$D$};
			\draw[-,gray,node distance=1cm] ($(D.north west)+(-0.2,0.8)$) rectangle
				($(E.south east)+(0.2,-0.3)$);
			\node[anchor=east] at ($(D.east)+(-0.7,0)$) {\scriptsize Non-dist. DEC~$D'$};
		\end{tikzpicture}
		\caption{MEMDP~$M$.}
		\label{fig:parity1}
	\end{subfigure}
	\begin{subfigure}[b]{0.39\textwidth}
		\centering
		\begin{tikzpicture}
			\tikzstyle{every state}=[node distance=1cm,minimum size=10pt, inner sep=1pt];
			\node[state] at (0,0) (A){$s_1$};
			\node[state,right of=A,node distance=0.7cm,fill=black,minimum size=3pt] (Aa) {};
			\node[state] at ($(Aa)+(0.7,0.8)$) (B){$s_D$};
			\node[state,node distance=0.7cm, right of=B, fill=black,minimum size=3pt] (Bb) {};
			\node[state,node distance=0.9cm, above right of=Bb] (C) {\scriptsize $W_D$};
			\node[state,node distance=1cm, below of=C] (LD) {\scriptsize $L_{D'}$};

			\node[state,node distance=1.3cm, below of=B] (D){$s_{D'}$};
			\node[state,node distance=0.7cm, right of=D, fill=black,minimum size=3pt] (Dd) {};
			\node[state,node distance=0.7cm, above right of=D, fill=black,minimum
			size=3pt] (Dld) {};
			\node[state,node distance=0.7cm, below right of=D, fill=black,minimum size=3pt] (De) {};
			\node[state,node distance=1cm, below of=D] (F) {$s_8$};
			\node[state,node distance=0.6cm, left of=F, fill=black, minimum size=3pt] (Ff) {};
			\node[state,node distance=1cm, left of=Ff] (G) {$s_9$};
			\node[state,node distance=1.8cm, right of=F] (Bot) {$s_5$};
			\node[state,node distance=0.5cm, above of=Bot] (Top) {$s_4$};

			\node at ($(Aa.north)+(0,0.2)$) {$a$};
			\node at ($(Bb.north)+(0,0.3)$) {$a_D^\$$};
			\node at ($(Dd.north)+(0,0.2)$) {$s_7b$};
			\node at ($(Dld.north)+(0,0.2)$) {$a_{D'}^\$$};
			\node at ($(De.north)+(0.35,-0.1)$) {$s_7c$};
			\node at ($(Ff.south)+(0,-0.2)$) {$a$};
			\path[-latex']
			(A) edge (Aa)
			(Aa) edge (B)
			(B) edge (Bb)
			(Bb) edge[bend right] (B)
			(Bb) edge  (C)
			(Aa) edge (D)
			(D) edge (Dld)
			(Dld) edge (LD)
			(D) edge (Dd)
			(D) edge (De)
			(Dd) edge[dashed] (Top)
			(Dd) edge[dotted] (Bot)
			(De) edge[bend left] (D)
			(De) edge (F)
			(C) edge[loop right]  (C)
			(LD) edge[loop right] (LD)
			(F) edge (Ff)
			(Ff) edge (D) 
			(Ff) edge[dotted] (G)
			(Top) edge[loop right] (Top)
			(Bot) edge[loop right] (Bot)
			(G) edge[loop left] (G);
			\draw[-,gray] ($(D.north west)+(-0.3,0.1)$) -- ($(D.north east)+(+0.2,0.1)$)
			-- ($(De.east)+(0.1,0)$)
			-- ($(F.south east)+(0.3,-0.3)$)
			-- ($(Ff.south west)+(-0.3,-0.3)$) -- cycle;
			\node[anchor=north] at ($(F.south)+(0,-0.2)$) {\scriptsize MEC~$D''$
				of~$M_1$};
		\end{tikzpicture}
		\caption{MEMDP~$\hat{M}$ where DECs were reduced.
		}
		\label{fig:parity2}
	\end{subfigure}
	\begin{subfigure}[b]{0.8\textwidth}
		\centering
		\begin{tikzpicture}
			\tikzstyle{every state}=[node distance=1cm,minimum size=10pt, inner sep=1pt];
			\node[state] at (0,0) (A){$s_1$};
			\node[state,right of=A,node distance=0.7cm,fill=black,minimum size=3pt] (Aa) {};
			\node[state,below of=A,node distance=0.7cm] (TD) {\scriptsize $t_{D''}^0$};
			\node[state] at ($(Aa)+(0.7,0.8)$) (B){$s_D$};
			\node[state,node distance=0.7cm, right of=B, fill=black,minimum size=3pt] (Bb) {};
			\node[state,node distance=0.9cm, above right of=Bb] (C) {\scriptsize $W_D$};
			\node[state,node distance=1cm, below of=C] (LD) {\scriptsize $L_{D'}$};

			\node[state,node distance=1.3cm, below of=B] (D){$s_{D'}$};
			\node[state,node distance=0.7cm, right of=D, fill=black,minimum size=3pt] (Dd) {};
			\node[state,node distance=0.7cm, above right of=D, fill=black,minimum
			size=3pt] (Dld) {};
			\node[state,node distance=0.7cm, below right of=D, fill=black,minimum size=3pt] (De) {};
			\node[state,node distance=1cm, below of=D] (F) {$s_8$};
			\node[state,node distance=0.6cm, left of=F, fill=black, minimum size=3pt] (Ff) {};
			\node[state,node distance=1cm, left of=Ff] (G) {$s_9$};
			\node[state,node distance=1.8cm, right of=F] (Bot) {$s_5$};
			\node[state,node distance=0.5cm, above of=Bot] (Top) {$s_4$};

			\node at ($(Aa.north)+(0,0.2)$) {$a$};
			\node at ($(Bb.north)+(0,0.3)$) {$a_D^\$$};
			\node at ($(Dd.north)+(0,0.2)$) {$s_7b$};
			\node at ($(Dld.north)+(0,0.2)$) {$a_{D'}^\$$};
			\node at ($(De.north)+(0.35,-0.1)$) {$s_7c$};
			\node at ($(Ff.south)+(0,-0.2)$) {$a$};

			\path[-latex']
			(A) edge (Aa)
			(Aa) edge (B)
			(B) edge (Bb)
			(Bb) edge[bend right] (B)
			(Bb) edge  (C)
			(Aa) edge (D)
			(D) edge (Dld)
			(Dld) edge (LD)
			(D) edge (Dd)
			(D) edge (De)
			(Dd) edge[dashed] (Top)
			(Dd) edge[dotted] (Bot)
			(De) edge[bend left] (D)
			(De) edge (F)
			(C) edge[loop right]  (C)
			(LD) edge[loop right] (LD)
			(F) edge (Ff)
			(Ff) edge (D) 
			(Ff) edge[dotted] (G)
			(Top) edge[loop right] (Top)
			(Bot) edge[loop right] (Bot)
			(G) edge[loop left] (G)
			(D) edge[dashed] node[below,pos=0.5]{\scriptsize $a_{D''}$}(TD)
			(F) edge[dashed] (TD)
			(TD) edge[loop left] (TD);
		\end{tikzpicture}
		\caption{The MEMDP~$\bar{M}$ obtained by adding the state~$t_{D''}^0$
			corresponding to the MEC~$D''$ of~$M_1$. The reachability
				objective is~$\bar{\Phi}=\objreach(\{s_4,s_9,t_{D''}^0,W_D\})$.}
		\label{fig:parity3}
	\end{subfigure}
	\caption{We are given an MEMDP in Fig.~\ref{fig:parity1} where we assume
		that~$D$ is a distinguishing DEC~$D$, and $D'$ is a non-distinguishing DEC
		(precise values of the probabilities do not matter). Distributions whose
		support differ in~$M_1$ and~$M_2$ are again shown in dashed or dotted lines.
		The parity function assigns~$0$ to~$s_4,s_9$ and~$w_D$, and~$1$ everywhere else.}
\end{figure}

\begin{lemma}
	\label{lemma:M-larger}
	Consider any MEMDP~$M$, and parity condition~$\calP$.
	For any state~$s$, strategy~$\bar{\sigma}$ for~$\bar{M}$,
	and~$\epsilon>0$, one can compute ${\sigma}$ such that
	$\pr_{{M}_i,s}^{\sigma}[\calP]\geq
	\pr_{\bar{M}_i,\bar{s}}^{\bar{\sigma}}[\bar{\Phi}]-\epsilon$
	for all~$i=1,2$.
\end{lemma}
\begin{proof}
	We define $\sigma$ as follows. For any history $h_1\ldots h_i \in
	\calH_\calT(\bar{M})$, such that
	$h_i\neq s_D$, and $a \in A(h_i)$, we let
	$\sigma(a \mid g) = \bar{\sigma}(a \mid h_1\ldots h_i)$,
	for any $g \in \red^{-1}(h_1\ldots h_i)$.
	For any $g \in \red^{-1}(h_1\ldots h_i s_D)$ with $s_D \in \mathcal{T}$,
	we run the strategy given by Lemma~\ref{lemma:strategy-for-mec-as} which
	achieves the objective with probability $1-\epsilon$.
	For any $g \in \red^{-1}(h_1\ldots h_i s_D)$ where $D$ is a non-distinguishing
	we switch to the strategy~$\tau$ of Lemma~\ref{lemma:strat-nddec} until~$D$ is
	left.
	Furthermore, at any history $g \in \red^{-1}(h_1\ldots h_i)$ with $h_i$
	belonging to some $D \in \calG(\bar{M}_j)$, we switch to the optimal
	strategy for~$M_j$ from~$g$ with probability $\bar{\sigma}(a_D \mid h_1\ldots h_i)$.

	We can then easily prove the following correspondance between histories
	of~$\bar{M}$ and that of~$M$. For any history $h_1\ldots h_i \in
	\calH_{\calT}(\hat{M})$, and action $a \in \hat{A}$,
	\begin{equation}
		\label{eqn:sigma-fromwt}
		\begin{array}{l}
				\pr_{\bar{M}_j,\bar{s}}^{\bar{\sigma}}(h_1\ldots h_i)=
				\pr_{M_j,s}^\sigma[\red^{-1}(h_1\ldots h_i) ],\\
				\pr_{\bar{M}_j,\bar{s}}^{\bar{\sigma}}(h_1\ldots h_ia)=
				\pr_{M_j,s}^\sigma[\red^{-1}(h_1\ldots h_ia) ],\\
				\pr_{\bar{M}_j,\bar{s}}^{\bar{\sigma}}[h_1\ldots h_i a_D^\$]=
				\pr_{M_j,s}^\sigma[\red^{-1}(h_1\ldots h_iD^\omega)].
		\end{array}
	\end{equation}
	Note that restricting histories to~$\hat{M}$ and~$\hat{A}$ only means that we exclude
	action~$a_D$ for MECs $D \in \calG(\bar{M}_j)$.

	The rest of the proof is done as in Lemma~\ref{lemma:wtM-larger}:
	We rewrite the probability of ensuring $\bar{\Phi}$ in~$\bar{M}$.
	\[
	\begin{array}{ll}
	\pr_{\bar{M}_j,\bar{s}}^{\bar{\sigma}'}[\bar{\phi}_j] &\geq 
		\sum_{H \in \calH_\calT\cdot \calT \setminus \bar{E}_1} \pr_{\bar{M}_j,\bar{s}}^{\bar{\sigma}'}
				[\bar{\phi}_j \mid H]
		\pr_{\bar{M}_j,\bar{s}}^{\bar{\sigma}'}[H]
		\\&~+\sum_{H \in \calH_\calT\setminus \bar{E}_1, D~\text{non-dist.}}
		\pr_{\bar{M}_j,\bar{s}}^{\bar{\sigma}'}[\bar{\phi}_j \mid H \cdot a_D^\$ (s_D^\$)^\omega]
		\pr_{\bar{M}_j,\bar{s}}^{\bar{\sigma}'}[H \cdot a_D^\$ (s_D^\$)^\omega]
		\\&~+\sum_{H \in \calD_K^j}
		\pr_{\bar{M}_j,\bar{s}}^{\bar{\sigma}'}[\bar{\phi}_j \mid H ]
				\pr_{M_j,s}^\sigma[H].
	\end{array}
	\]
	and show that conditioned on each above event, the probability of winning
	in~$M$ is at least as high as the expectation in~$\bar{M}$, up to~$\epsilon$.
	On histories $H \in \calH_\calT\cdot \calT$ this follows by
	Lemma~\ref{lemma:strategy-for-mec-as}.
	On histories that end with $D^\omega$ for
	non-distinguishing components~$D$, the probability of winning in~$\bar{M}_j$
	is the optimal probability of winning in~$M_j$ from~$D$, which is achieved
	by~$\sigma$.
	Last, the probability of winning conditioned on 
	$\calD_K^j$ is equal to the optimal probability of winning
	in~$M_j$ from the current state by construction, and this is the probability
	achieved from such histories in~$M$ by definition of~$\sigma$.
\end{proof}

We summarize the result we proved in the following theorem.

\begin{theorem}
	\label{thm:quantparity}
	The quantitative parity problem is polynomial-time reducible to the
	quantitative reachability problem. The limit-sure parity problem is 
	in polynomial time.
\end{theorem}

\bibliographystyle{myabbrv}
\bibliography{biblio}
\end{document}